\documentclass[11pt, reqno]{amsart}
\usepackage[a4paper,margin=2.5cm]{geometry}

\usepackage{amsmath,amssymb,amsfonts,amscd,mathtools}
\usepackage[T1]{fontenc}
\usepackage{graphicx}
\usepackage{booktabs}
\usepackage{array}
\usepackage{bbm}
\usepackage{bm}
\usepackage{dsfont}
\usepackage{mathrsfs}
\usepackage{empheq}
\usepackage{latexsym}
\usepackage{verbatim}
\usepackage{cite}
\usepackage{xcolor}
\usepackage{caption}
\usepackage{subcaption}
\usepackage{soul}
\usepackage{ytableau}
\usepackage{youngtab}
\usepackage{braket}
\usepackage{slashed}
\usepackage{tcolorbox}
\tcbuselibrary{theorems,skins}

\usepackage{tikz}
\usepackage{tikz-cd}
\usetikzlibrary{matrix,shapes.misc,positioning,decorations.pathreplacing,
  angles,quotes,decorations.pathmorphing}

\usepackage{hyperref}
\hypersetup{
  colorlinks=true,
  linkcolor=blue,
  citecolor=blue,
  urlcolor=blue,
  filecolor=magenta,
  linktocpage=true
}
\usepackage[nameinlink,noabbrev]{cleveref}
\crefname{figure}{Fig.}{Figs.}
\Crefname{figure}{Fig.}{Figs.}
\crefname{equation}{Eq.}{Eqs.}
\Crefname{equation}{Eq.}{Eqs.}
\crefname{section}{Section}{Sections}
\Crefname{section}{Section}{Sections}

\numberwithin{equation}{section}
\renewcommand{\Re}{\mathrm{Re}}
\renewcommand{\Im}{\mathrm{Im}}
\renewcommand{\sl}{\sqrt{\lambda}}
\newcommand{\ev}[1]{\big\langle #1 \big\rangle}
\newcommand{\dd}{\mathrm{d}}

\definecolor{bleudefrance}{rgb}{0.19,0.55,0.91}

\def\cO{\mathcal{O}}
\newcommand{\be}{\begin{equation}}
\newcommand{\ee}{\end{equation}}
\newcommand{\nn}{\nonumber}

\newcommand{\Nf}{\mathcal{N}=4}

\newcommand{\F}{\mathcal{F}}

\newcommand{\D}{\mathcal{D}}

\newcommand{\lb}{\left(}
\newcommand{\rb}{\right)}

\newcommand{\opchi}{\boldsymbol\chi}

\newcommand{\Lo}{L^{\rm odd}}
\newcommand{\Le}{L^{\rm even}}
\DeclareMathOperator{\Tr}{Tr}

\renewcommand{\leq}{\leqslant}
\renewcommand{\geq}{\geqslant}

\newcommand{\QQ}{\mathbbm{Q}}

\newtheorem{thm}{Theorem}[section]

\newtheorem{lem}[thm]{Lemma}

\title[Giant Wilson loops]{Giant Wilson loops from Fredholm determinants}
\author{Ludvig Chen}
\email{ludvig.chen@physics.uu.se}
\author{Victor Mishnyakov}
\email{mishnyakovvv@gmail.com}
\author{Konstantin Zarembo}
\email{zarembo@nordita.org}

\begin{document}
\begin{abstract}
We study large-$k$ symmetric and antisymmetric Wilson loops in $\mathcal{N}=4$ super-Yang-Mills theory. We represent their generating functions as Fredholm determinants built from the Hermite kernel weighted by a symbol function that encodes the representation. Using the integrable structure of the kernel, we systematically apply the Tracy-Widom approach to derive a closed system of integro-differential equations. The equations are then solved order by order in $1/N$ by a WKB ansatz, while keeping the symbol function general. We find a closed-form expression for the leading and subleading order contributions to the generating functions. Finally, specializing to the two representations, we perform the large-$k$ saddle-point analysis and find the subleading contribution to the Wilson loop. Our results confirm the previously observed discrepancies with the dual D-brane description of giant Wilson loops.
\end{abstract}

\maketitle
\tableofcontents
\pagestyle{plain}

\section{Introduction}

The circular Wilson loop in the $\mathcal{N}=4$ super-Yang-Mills theory (SYM) is
an exceptional observable with remarkable mathematical properties. Its expectation value is described by the Gaussian matrix integral \cite{Erickson:2000af,Drukker:2000rr}, by virtue of supersymmetric localization \cite{Pestun:2007rz}. As such it can be computed exactly \cite{Drukker:2000rr} or in the large-$N$ limit \cite{Erickson:2000af} and interpolated from weak to strong coupling while keeping a full control over the intermediate regime, offering thereby a concrete realization of string worldsheet emergence from a sum of infinitely many planar diagrams. 

An extension to giant Wilson loops, namely Wilson loops in higher representations whose rank scales linearly with $N$, gives access to all orders in the $1/N$ expansion. These observables are dual to D-branes and probe the AdS/CFT correspondence beyond the planar limit. A truly spectacular agreement observed at the classical level \cite{Drukker:2005kx, Yamaguchi:2006tq, Hartnoll:2006is} unfortunately does not extend to the quantum domain. An explicit account for D-brane fluctuations \cite{Faraggi:2011ge, Buchbinder:2014nia,Faraggi:2014tna} has remained in persistent disagreement with the predictions of the matrix model \cite{Chen-Lin:2016kkk, Gordon:2017dvy}. In one dimension lower, in ABJM, the M2-brane fluctuations reproduce $1/N$ corrections to Wilson loops perfectly well \cite{Giombi:2023vzu}. In light of this success and the existing discrepancies it is important to reassess the 4d calculations on both sides. In this paper we revisit the matrix-model calculation of giant Wilson loops in the rank-$k$ symmetric and antisymmetric representations.

The key technical novelty of our approach is the systematic application of the Tracy--Widom method \cite{Tracy:1993ah}, traditionally used to characterize eigenvalue correlations in random matrix theory \cite{Tracy:1992rf}. The same framework can be adapted to compute and analyze Wilson-loop expectation values, opening a new route to the matrix-model description of giant Wilson loops. 

The starting point of this method is a compact determinant representation for the generating function of antisymmetric  Wilson loops  \cite{Fiol:2013hna} and its extension to the symmetric case  \cite{Okuyama:2018aij} enabled by a remarkable connection to the truncated harmonic oscillator   \cite{Fiol:2013hna}. The oscillator algebra and its coordinate-space wavefunctions have already been used to set up the large-$N$ expansion of the generating function \cite{CanazasGaray:2018cpk}. We go one step further and replace the $N\times N$ determinant by a Fredholm determinant of an integral operator with the Hermite kernel. What might look as an unnecessary complication eliminates the cumbersome finite-$N$ projection of the Hilbert space and opens an avenue for conventional WKB methods at large $N$.  It also puts symmetric and antisymmetric representations on an equal footing allowing us to treat them in parallel. More importantly, the Fredholm representation embeds the formalism into the general theory of integral operators with integrable kernels \cite{1990IJMPB...4.1003I, Tracy:1993ah}. Integrable Fredholm determinants arise in other observables of the $\mathcal{N}=4$ SYM \cite{Belitsky:2019fan,Belitsky:2020qrm}, and computational techniques developed in this context were an important inspiration for our work, even though technical details differ when the method is applied to Wilson loops.

\subsection{Summary of the results}

We consider the locally BPS Wilson loop of the $\mathcal{N}=4$ SYM in a representation $R$ of $U(N)$:
\begin{equation}
 W_R(C)=\left\langle\mathrm{Tr}_R\,
 P\exp\oint\limits_C \dd s\,\left(i\dot{x}^\mu A_\mu +|\dot{x}|n^i\Phi _i\right)
 \right\rangle_{\rm SYM}.
\end{equation}
The spacial contour $C$ has a circular shape and  $n^i$ is a constant six-dimensional vector of the unit norm. The path integral for the circular loop localizes on the scalar zero modes \cite{Pestun:2007rz} and reduces to an expectation value in the 
Gaussian matrix model 
\cite{Erickson:2000af,Drukker:2000rr}:
\begin{equation}
     W_R(N,\lambda) = \big\langle \Tr_{R} e^X \big\rangle.
\end{equation}
The bracket now denotes an average over constant Hermitian  $N\times N$ matrices:
\begin{equation}
    \ev{f(X)} = \frac1{Z_0}\int \dd X \exp\left(- \dfrac{2N}{\lambda}\Tr X ^2 \right) f(X).
\end{equation}
Here $\lambda=g_{\rm YM}^2N$ is the SYM 't Hooft coupling, and the matrix model expectation value is normalized by the partition function
\begin{equation}
    Z_0= \int \dd X \exp\left(- \dfrac{2N}{\lambda}\Tr X ^2 \right),
\end{equation}
ensuring that $W_{\varnothing}=1$.

Importantly, localization works for Wilson loops in any representation. The representation label provides an additional parameter and facilitates connection to
D-branes  \cite{Gomis:2006sb,Gomis:2006im}. Indeed, it is the D-branes that furnish an appropriate holographic description of the giant Wilson loops considered here \cite{Drukker:2005kx}. The giant loops belong to parametrically large representations for which the conventional ’t~Hooft counting breaks down, necessitating specialized techniques for computing expectation values at large $N$.

In this paper we focus on Wilson loops in the rank-$k$ symmetric $S_k$
and antisymmetric $A_k$ representations of $U(N)$\footnote{
There is a simple relation between 
the $U(N)$ and $SU(N)$ expectation values. For a Wilson loop in the representation $R$ with  $k$ boxes in the Young diagram, $W^{SU(N)}_R = e^{-\frac{\lambda k^2}{8N^2}}W^{U(N)}_R$, see \cite{Drukker:2000rr} and \cite{Gordon:2017dvy}.},
which correspond to the Young diagrams 
\begin{equation}
S_k = \begin{ytableau}
~ & ~ & ~ & \none[\dots] & ~
\end{ytableau}
\hspace{-8em}
\underbrace{\hspace{8em}}_{\text{\normalsize $k$ boxes}}\qquad \text{and }\qquad
A_k = \begin{tabular}{c}
\begin{ytableau}
~ \\ ~ \\ ~ \\ \none[\vdots] \\ ~
\end{ytableau}
\end{tabular}
\hspace{-1.5em}
\left.\begin{array}{c}
\rule{0pt}{2em} \\ \rule{0pt}{2em} \\ \rule{0pt}{2.8em}
\end{array}\right\} k \text{ boxes}.
\end{equation}
For the antisymmetric representation, $k$ is bounded above by $N$, whereas the symmetric representations have no such restriction. We will be interested in the simultaneous limit of $N\rightarrow\infty$ and $k\rightarrow\infty$ with their ratio held fixed:
\begin{equation}\label{eq:f}
    f=\dfrac{k}{N}.
\end{equation}
This regime characterizes the giant Wilson loops.

Instead of dealing with individual representations it is convenient to define the generating functions:
\begin{align}\label{gen-functions-intro}
    F^{S}(t) = \sum_{k=0}^\infty t^k \Tr_{S_k} e^X,
    \qquad 
    F^{A}(t) = \sum_{k=0}^N t^k \Tr_{A_k} e^X.   
\end{align}
A Wilson loop in a given representation can then be extracted by a contour integral:
\begin{equation}
    W_{A_k,S_k}=\frac{1}{2\pi i}\oint \dfrac{\dd t}{t^{k+1}} \ev{F^{A,S}(t)}.
    \label{eq:Wilsonloop-contour-integral-intro}
\end{equation} 
The expectation values of the generating functions exponentiate at large $N$:
\begin{equation}   \label{eq:free-energy-definition}
 \ev{F^{A,S}(t)}=\,{\rm e}\,^{N  \mathcal{F}^{A,S}(t)},
\end{equation}
where the free energy $\mathcal{F}^{A,S}(t)$ admits a systematic 
large-$N$ expansion:
\begin{equation}
   \mathcal{F}^{A,S}(t)=\sum_{n=0}^\infty \mathcal{F}^{A,S}_n(t) N^{-n}.
\end{equation}
As a consequence of exponentiation, the contour integral in (\ref{eq:Wilsonloop-contour-integral-intro}) naturally picks a saddle point at $t\sim 1$ when $N$ is large and $k\sim N$. Taking $N\rightarrow \infty $ at fixed $t$ is therefore an equivalent definition of the giant Wilson loop regime.

We will develop a systematic framework to calculate $\mathcal{F}^{A,S}(t) $ order by order in $1/N$. As an illustration we explicitly evaluated the first two orders:
\begin{align}
    &\F^{A,S}_0(t) = \pm\frac{2}{\pi}\int\limits_{-1}^1\dd x\,\sqrt{1-x^2}\,
    \log\big(1\pm t \,{\rm e}\,^{\sqrt{\lambda}x}\big), 
    \\
    &\F_1^{A,S}(t) = \frac{\sqrt{\lambda}}{2\pi^2}\int\limits_{-1}^1\frac{\dd x}{\sqrt{1-x^2}}\int\limits_{-1}^1 \frac{y \, \dd y}{\sqrt{1-y^2}} \frac{\,{\rm e}\,^{\sqrt{\lambda}x}}{\,{\rm e}\,^{\sqrt{\lambda}y}-\,{\rm e}\,^{\sqrt{\lambda}x}}\log \frac{1\pm t \,{\rm e}\,^{\sqrt{\lambda}x}}{1\pm t\,{\rm e}\,^{\sqrt{\lambda}y}} \,,
\end{align}
where the upper and lower signs correspond to the antisymmetric and symmetric representations respectively. These formulas reproduce the well-known leading-order results for symmetric and antisymmetric Wilson loops \cite{Hartnoll:2006is} and the first correction for the antisymmetric representations \cite{Gordon:2017dvy,Okuyama:2017feo,CanazasGaray:2018cpk}. They are also consistent with the symmetry relation between symmetric and antisymmetric Wilson loops conjectured in \cite{Okuyama:2018aij}.

The generating functions above are readily projected onto the fixed-rank representations. The saddle-point evaluation of the Mellin integral in (\ref{eq:Wilsonloop-contour-integral-intro}) gives for the symmetric and antisymmetric Wilson loops at strong coupling ($\lambda \rightarrow \infty $) and to the first two orders in $1/N$:
\begin{align}
    &\log W_{A_k} = \frac{2N \sqrt{\lambda}}{3\pi}\sin^3\theta+\dfrac{\lambda}{8\pi^2}\left( \sin^2 \theta-\theta \sin2\theta+\theta^2 
    \right) +\ldots, \qquad 
    \theta - \cos \theta \sin \theta=\pi f
    \\
     &\log W_{S_k} = 2N\lb \kappa \sqrt{1+\kappa^2}+ {\rm arcsinh}\, \kappa \rb-\frac12 \log\Big(\kappa^3 \sqrt{1+\kappa^2}\Big) +\ldots,
     \qquad 
     \kappa =\frac{\sqrt{\lambda }\,f}{4}
\end{align}
The subleading corrections reproduce the earlier matrix-model calculations \cite{Gordon:2017dvy,Chen-Lin:2016kkk}\footnote{The result for the symmetric representations \cite{Chen-Lin:2016kkk} rests on an additional assumption that the symmetric Wilson loops are equivalent to $k$-fundamental ones at strong coupling \cite{Drukker:2005kx,Yamaguchi:2007ps}. This can be justified by a contour argument \cite{Fraser:2011qa,Chen-Lin:2015dfa}, which is water-tight at the leading order in $1/N$ but is a conjecture beyond that. Our calculation deals directly with the symmetric representations and does not rely on this assumption.} and therefore disagree with the analyses of quantum fluctuations of the D5 and D3 branes \cite{Buchbinder:2014nia,Faraggi:2014tna}.

\section{Wilson loops in higher representations}
\label{sec:wilsonloops}

Localization of the path integral maps the Wilson loop insertions to characters of the exponential field in the matrix model. For a given representation $R$, which we specify by its Young diagram, the character  $\Tr_R e^X$ is a gauge-invariant function of the matrix $X$ or, in  other words, is a symmetric function of its eigenvalues. The characters can be re-expressed through multi-traces,
for example:
\begin{equation}
    \Tr_{[2]} e^X  = \dfrac{1}{2}\left( \Tr e^{2X}+\left(\Tr e^{X} \right)^2\right).
\end{equation}
The character of an irreducible finite dimensional representation of $U(N)$ is a Schur polynomial of the eigenvalues, in our case:
\begin{equation}
    \Tr_{R} e^{X}=s_{R}\left(e^{x_1},\ldots, e^{x_N}\right).
\end{equation}
The characters of rank-$k$ symmetric and antisymmetric representations have a particularly simple form, they are the complete homogeneous symmetric polynomials $h_k$ and the elementary symmetric polynomials $e_k$ of the eigenvalues:
\begin{align}
    \Tr_{S_k} e^X &= h_k(e^{x_1},\dots,e^{x_N})
      = \sum_{1\leq i_1\leq\cdots\leq i_k\leq N} e^{x_{i_1}+\cdots+x_{i_k}}, \nonumber\\
    \Tr_{A_k} e^X &= e_k(e^{x_1},\dots,e^{x_N})
      = \sum_{1\leq i_1<\cdots<i_k\leq N} e^{x_{i_1}+\cdots+x_{i_k}}.
      \label{symmetricpolynomialDEF}
\end{align}
Our primary objective is to compute their expectation values in the Gaussian matrix model:
\begin{equation}
    W_{A_k,S_k}(N,\lambda)= \ev{\Tr_{A_k,S_k}e^{X}}.
\end{equation}

For a gauge-invariant observable like this, the angular degrees of freedom can be integrated out leaving behind an eigenvalue integral:
\begin{equation}
   \ev{f(x_1,\ldots x_N)} = \frac{1}{Z_0} \int\limits_{-\infty}^{{\infty}} \prod_{i=1}^{N} 
   \dd 
   x_i \,\Delta^2(x) \exp\left( -\dfrac{2N}{\lambda}\sum_{i=1}^N x_i^2 \right)f(x_1,\ldots x_N),
\end{equation}
where $\Delta(x)$ is the Vandermonde determinant.

The individual characters can be organized into generating functions (\ref{gen-functions-intro}), which take a particularly simple form for symmetric and antisymmetric representations:
\begin{align}
    F^{S}(t) &
      = \prod_{i=1}^N \frac{1}{1-te^{x_i}}
      = \frac{1}{\det(1-te^X)}\,, \label{generatingfunctionsBoth_sym}\\
    F^{A}(t) &
      = \prod_{i=1}^N (1+te^{x_i})
      = \det(1+te^X). \label{generatingfunctionsBoth_anti}
\end{align}
Our main focus will be on the expectation values $\left\langle F^{A,S}(t)\right\rangle$ in the large-$N$ limit. We note, however, that several interesting results for Wilson loops and their generating functions have been obtained at finite $N$  without assuming planarity
\cite{Fiol:2013hna,Okuyama:2018aij,Fiol:2018yuc,CanazasGaray:2019mgq,Morozov:2026rga}. Giant Wilson loops with local operator insertions have been also studied \cite{Giombi:2020amn}. 

There is subtlety in defining the generating function for the symmetric characters. One can see already at $N=1$ that the series:
\begin{equation}
    \ev{F^{S}(t)}\Big|_{N=1}=\sum_{k=0}^{\infty}t^k e^{\frac{k^2\lambda}{8}} 
\end{equation} 
has a vanishing radius of convergence. This can be understood from the right-hand side of \eqref{generatingfunctionsBoth_sym}. At finite $N$, the eigenvalues are integrated over the entire real axis, so for any positive real $t$ the integration will hit the pole at $x=-\log t$. At large $N$ the singularity is not as severe as the eigenvalues condense onto a cut. The cut in the spectral plane will induce a branch cut in the generating function in the $t$-plane.

The contour integral  \eqref{eq:Wilsonloop-contour-integral-intro} that projects the generating function onto individual representations becomes an inverse Laplace transform in the variable $\log t$:
\begin{equation}
    W_{A_k, S_k} = \oint \frac{\dd t}{2\pi it}\,\,
    e^{N\left(\F^{A,S}(t)-f\log t\right)}.
\end{equation}
In the large-$N$ limit with $f$ fixed the integral is naturally saturated by a saddle point.

\subsection{Exact determinant representation}
\label{sec:exact-determinant-representation}
The first key ingredient to our construction is the exact determinant representation for the expectation values of generating functions \eqref{generatingfunctionsBoth_sym} and \eqref{generatingfunctionsBoth_anti} at finite values of $N$. These representations were derived using the method of orthogonal polynomials \cite{Fiol:2013hna,Okuyama:2018aij}. Below, we discuss these in detail since they are crucial for our next step.

The expectation value of the antisymmetric generating function was computed
exactly in \cite{Fiol:2013hna}. In this case it takes the following form
\begin{equation}
    \langle F^A(t)\rangle
    = \frac{1}{Z_0}\int \prod_{i=1}^N \dd x_i\,\Delta^2(x)\,F^A(t)
      \exp \left(-\frac{2N}{\lambda}\sum_{j=1}^N x_j^2\right)
    = \det \Big(1+ M^A\Big), \label{exactgeneratingfunctionformulaA_k}
\end{equation}
where $M^A$ is an $N\times N$ matrix with elements
\begin{align}
    M^A_{ij}
    &=
    \int\limits_{-\infty}^{\infty}\dd x\,
    \varphi_i(x)\varphi_j(x)e^{-2Nx^2}\,
    t e^{\sqrt{\lambda}x} =
    t e^{\frac{\lambda}{8N}}
    \sqrt{\frac{i!}{j!}}
    \left(\sqrt{\frac{\lambda}{4N}}\right)^{j-i}
    L_i^{\,j-i}\!\left(-\frac{\lambda}{4N}\right),
    \label{MatrixelementLaguerre}
\end{align}
for $0\leq i,j \leq N-1$. We redefined the integration variable as $x\rightarrow x\sqrt{\lambda}$ for later convenience. 

Here $\varphi_i(x)$ are the Hermite polynomials that form the orthogonal polynomial basis of the Gaussian matrix model and are normalized with respect to the weight $e^{-2Nx^2}$:
\begin{equation}
    \int\limits_{-\infty}^\infty \dd x\,\varphi_i(x)\varphi_j(x)\,e^{-2Nx^2} = \delta_{ij}, \qquad
    \varphi_i(x) = \left(\frac{2N}{\pi}\right)^{1/4}\frac{1}{\sqrt{2^i i!}}\,H_i(\sqrt{2N}\,x).
\end{equation}
Absorbing the weight into the definition of the wavefunction: $\psi_n(x) = e^{-Nx^2}\varphi_n(x)$, we arrive at the eigenstates of the harmonic oscillator:
\begin{equation}
    -\frac{1}{4N^2}\frac{\dd^2\psi_n}{\dd x^2} + x^2\psi_n
    = \frac{1}{N} \left(n+\frac{1}{2}\right) \psi_n.
\end{equation}
Here $\hbar=1/N$ plays the role of the Planck constant, which suggests that the $1/N$ expansion should be equivalent to the WKB expansion of the quantum-mechanical matrix elements of the harmonic oscillator. This is essentially the idea we pursue in this paper.

If we denote by $\ket{\psi_n}$ the respective oscillator states, then $M_{ij}^{A}$ are matrix elements of the momentum-shift operator $te^{\sqrt{\lambda}\hat{x}}$:
\begin{equation}
    M_{ij}^{A}=\bra{\psi_i} te^{\sqrt{\lambda}\hat{x}} \ket{\psi_j}.
\end{equation}
Note, however, that the determinant \eqref{exactgeneratingfunctionformulaA_k} only selects the lowest $N$ eigenstates. The mapping to the harmonic oscillator described above was used for computing $1/N$ corrections in  \cite{CanazasGaray:2018cpk}, where the oscillator Hilbert space was truncated to the first $N$ levels to match the size of the matrix $M_{ij}^A$. We will develop this idea in an alternative direction, where instead of truncating the Hilbert space we will extend the matrix to an operator and represent the generating function as a  Fredholm determinant. A seemingly different but related Fredholm-determinant representation was derived in \cite{Giombi:2020amn}.

Following the same orthogonal-polynomial approach as in the antisymmetric representation, the generating function for the symmetric representation can be represented as
\begin{align}
    \langle F^S(t)\rangle
    &= \det_{0\leq i,j \leq N-1} \left( \int_{-\infty}^{\infty} \dd x\,\varphi_i(x)\varphi_j(x) \dfrac{1}{1-t e^{\sqrt{\lambda}x}}e^{-2Nx^2} \right).
\end{align}
Using the orthonormality of $\varphi_i(x)$, we can isolate the identity matrix and bring this expression to the same form as the antisymmetric generating function \eqref{exactgeneratingfunctionformulaA_k}
\begin{equation}
    \langle F^S(t)\rangle = \det \left(1+M^S\right), \label{determinantSymmetric}
\end{equation}
with the matrix $M^{S}_{ij}$, for $0\leq i,j \leq N-1$,
\begin{equation}
    M^S_{ij} = \int_{-\infty}^\infty \dd x\,\varphi_i(x)\varphi_j(x)\,e^{-2Nx^2}
               \frac{te^{\sqrt{\lambda}x}}{1-te^{\sqrt{\lambda}x}}.
               \label{SmatrixIntegralform}
\end{equation}
Unlike $M^A_{ij}$, which evaluates to a closed form in terms of generalized Laguerre polynomials \eqref{MatrixelementLaguerre}, the matrix elements $M^S_{ij}$ do not admit a simple closed-form expression. Expanding the integrand in powers of $t$ instead expresses the result as an infinite sum of $M^A_{ij}$ matrices. 

Let $\widetilde{M}^{A}_{ij}(n^2\lambda)$ denote the matrix \eqref{MatrixelementLaguerre} with the factor of $t$ removed:
\begin{equation}
    \widetilde{M}^A_{ij}(n^2\lambda) = \exp\left(\frac{n^2\lambda}{8N}\right) \sqrt{\frac{ i!}{ j!}} \left(\sqrt{\frac{n^2\lambda}{4N}}\right)^{j-i} L^{j-i}_{i}\left(-\frac{n^2\lambda}{4N}\right).
\end{equation}
The symmetric generating function can then be written as \cite{Okuyama:2018aij}
\begin{equation}
       \langle F^S(t)\rangle  = \det \lb 1+\sum_{n=1}^\infty t^n \widetilde{M}^A(n^2\lambda)\rb.
\end{equation}
This expanded representation was studied in \cite{Okuyama:2018aij}, where it was used to obtain many orders of the weak coupling expansion. One could in principle study the symmetric Wilson loop at finite $N$ using this form, but for our purposes, the unexpanded integral form \eqref{determinantSymmetric} is more useful.

\subsection{Symmetry relations}
\label{sec:symmetry-relations}
The expectation values of the symmetric and antisymmetric generating functions are related by a symmetry. Namely, for individual Wilson loops at finite $N$ one expects the following property to hold at least to all orders in perturbation theory in $\lambda$ \cite{Fiol:2018yuc,Muck:2019hnz}:
\begin{equation}
    W_{S_k}(N,\lambda)=(-1)^{k}W_{A_k}(-N,\lambda).
\end{equation}
At the level of generating functions this translates to:
\begin{equation}
    \langle F^S(t)\rangle =\langle F^A(-t)\rangle\Big|_{N\rightarrow-N}.
\end{equation}
Although the substitution  $N \rightarrow-N$ cannot be defined for the determinants themselves, it is well defined in the large-$N$ expansion
\begin{equation}\label{eq:StoAproperty}
    \mathcal{F}^{A}(t,N)=-\mathcal{F}^{S}(-t,-N).
\end{equation}
For the coefficients in the free energy expansion, this implies
\begin{equation}
    \F_n^A(t) =(-1)^{n+1}\F_n^S(-t).
\end{equation}
The symmetry holds at finite $N$ but it could in principle be violated in the large-$N$ limit. However, it was observed in \cite{Okuyama:2018aij} that \eqref{eq:StoAproperty} holds in the weak coupling and $1/N$ expansion. It was further established perturbatively in $\lambda$ in \cite{Fiol:2018yuc} and subsequently at finite $\lambda$ in \cite{CanazasGaray:2019mgq} to all orders in $1/N$. Our calculation confirms this symmetry for generic values of $\lambda$ at the leading and subleading order in $N^{-1}$.

The antisymmetric generating function has an additional symmetry, due to the symmetry of the Wilson loops with respect to taking conjugate representations: 
\begin{equation}
    \ev{ \Tr_R e^{X} }= \ev{ \Tr_{\bar{R}} e^{X} } e^{\frac{\lambda}{8N^2}
    \left(  |R|^2-|\bar{R}|^2\right)}
\end{equation}
where $|R|=\sum\limits_{i=1}^N R_{i}$.
For antisymmetric representations this implies:
\begin{equation}
    W_{A_k}(N,\lambda)=W_{A_{N-k}}(N,\lambda) e^{\frac{\lambda}{8N}(2k-N)}.
\end{equation}
The exponential factor is a feature in passing from $SU(N)$, for which the expectation values of conjugate representations are equal, to $U(N)$, where the additional $U(1)$ contribution must be taken into account. This symmetry is reflected in the large-$N$ limit, where $k \rightarrow N-k$ corresponds to transforming their ratio \eqref{eq:f} as $f \rightarrow1-f$.

\section{Wilson loops as Fredholm determinants}
\label{sec:fredholm}
In this section, we develop a new approach for computing symmetric and antisymmetric Wilson loops using the formalism of Fredholm determinants. As we will see, this formalism provides a systematic framework for computing higher order $1/N$ corrections.

Our approach is inspired by the Tracy-Widom analysis of the eigenvalue distribution in random matrices, which relates
eigenvalue statistics in Hermitian random matrix models to Fredholm
determinants with completely integrable kernels \cite{Tracy:1993ah,Tracy:1992rf,Tracy:1993xj,Tracy:1992mk,Meh2004}.
The method of differential equations \cite{Tracy:1993ah,1990IJMPB...4.1003I,Korepin:1993kvr}
then yields closed systems of differential equations for these determinants.  A number of observables in $\mathcal{N}=2$ and $\Nf$ SYM have been computed along these
lines \cite{Belitsky:2020qir,Beccaria:2022ypy,Bajnok:2024bqr,Bajnok:2024epf,Bajnok:2024ymr}.
Specific examples include the four-point function of heavy $\tfrac12$-BPS
operators \cite{Belitsky:2019fan}, flux-tube correlations
\cite{Belitsky:2019fan,Beisert:2006ez,Basso:2020xts} and non-planar
corrections to the free energy on $S^4$ \cite{Beccaria:2021vuc,Beccaria:2023kbl}.

More recently, related Fredholm determinant methods have been used to study the strong coupling and transseries structure for $\Nf$ SYM observables described by matrix Bessel kernels \cite{Bajnok:2025lji,Boldis:2026rkb,Bajnok:2026xri,Boldis:2026dcq}.

\subsection{Fredholm representation}
The result of the previous section can be unified by 
introducing the so-called \emph{symbol functions}%
\footnote{Also called cut-off functions in the literature.}: 
\begin{equation}
\label{eq:sym-antisym-symbol-functions}
    \chi^S(x,t) =  \frac{te^{\sqrt{\lambda}x}}{1-te^{\sqrt{\lambda}x}}, \qquad
    \chi^A(x,t) = te^{\sqrt{\lambda}x}\,.
\end{equation}
For a generic symbol function $\chi(x,t)$, the matrix $ M[\chi]$ is defined by the matrix elements of the operator $\chi(x,t)$ in the basis of the harmonic oscillator eigenfunctions:
\begin{equation}
    M[\chi]_{ij} = \int_{-\infty}^{\infty} \dd x\,\varphi_i(x)\varphi_j(x)\,e^{-2Nx^2}\,\chi(x,t)  = \bra{\psi_i}\chi(\hat{x},t)\ket{\psi_j
    }.
\end{equation}
In terms of these symbol functions, the exact determinant representations for generating functions \eqref{exactgeneratingfunctionformulaA_k} and \eqref{determinantSymmetric} can be written as
\begin{equation}
  \begin{split}
        \ev{F^{A,S}(t)}=\det\left(1+  M\left[\chi^{A,S}\right] \right),
  \end{split}
\end{equation}
where the only difference between the two representations is encoded in the choice of symbol function. Below we will drop the explicit $\chi$ dependence and simply denote the matrix $M[\chi]$ as $M$. Moreover, our analysis is agnostic to a precise shape of the symbol function and will only use its homogeneity properties.

The matrix $M$ appearing in the determinants can be thought of as the matrix representation of the quantum mechanical operator $\chi(\hat{x},t)$ in the basis of harmonic oscillator eigenstates. In addition, the operator is truncated to the first $N$ levels. It is this Hilbert space truncation that makes the large-$N$ limit tricky. The key idea is to reformulate the finite-$N$ truncation as an insertion of a projector
\begin{equation}
\hat{P}_N=\sum_{j=0}^{N-1} \ket{\psi_j}\bra{\psi_j}.
\end{equation}
The original matrix $M$ then becomes
\begin{equation}\label{eq:ProjectorIdentity1}
    \bra{\psi_i} \hat{P}_N\chi(\hat{x},t)\hat{P}_N \ket{\psi_j}= \left\{\begin{split}
        &M_{ij}\, \quad 0\leq i\leq N-1, \quad 0\leq j\leq N-1
        \\
        &0 \,\quad i\geq N \text{ or } j\geq N.
    \end{split} \right.
\end{equation}
Equivalently, directly in position space, the projector can be represented as an integral operator with a kernel given by:
\begin{equation}
\begin{split}
    H(x,y) &=\bra{x}\hat{P}_N \ket{y}=  e^{-Nx^2-Ny^2}\sum_{j=0}^{N-1}\varphi_j(x)\varphi_j(y)\\
    &= \frac{e^{-Nx^2-Ny^2}}{\sqrt{\pi}\,2^N(N-1)!}
      \frac{H_N(\sqrt{2N}x)H_{N-1}(\sqrt{2N}y)-H_{N-1}(\sqrt{2N}x)H_N(\sqrt{2N}y)}{x-y}.
      \label{christoffeldarbouxkernel}
    \end{split}
\end{equation}
Here the second equality follows from the well-known Christoffel-Darboux identity. The kernel $ H(x,y)$ explicitly depends on $N$, but for notational convenience we suppress this dependence throughout. Similarly, we suppress the $N$ dependence of the Hermite polynomials appearing above by defining\footnote{We note in passing that $\phi=\psi_N/\sqrt{2}$ and is therefore normalized to $1/\sqrt{2}$.\label{footnote4}}
\begin{equation}
    \phi(x) = \alpha_N e^{-Nx^2}H_N(\sqrt{2N}x), \qquad
    \psi(x) = \beta_N e^{-Nx^2}H_{N-1}(\sqrt{2N}x),
    \label{eq:phi-psi-definition}
\end{equation}
with $\alpha_N = (2N)^{-1/4}\Big(\sqrt{\pi}\,2^N(N-1)!\Big)^{-1/2}$ and
$\beta_N = (2N)^{1/4}\Big(\sqrt{\pi}\,2^N(N-1)!\Big)^{-1/2}$. 
In this condensed notation the kernel takes the form
\begin{equation}
    H(x,y) = \frac{\phi(x)\psi(y)-\psi(x)\phi(y)}{x-y}. \label{hermitkernel}
\end{equation}
Then \eqref{eq:ProjectorIdentity1} can be written in the position space  as
\begin{equation}
    M_{ij}=\int\limits_{-\infty}^{\infty}\dd x \int\limits_{-\infty}^{\infty}\dd y
 \,\psi_i(x) \psi_j(y) H(x,y) \chi(x,t), \quad 0\leq i,j \leq N-1.
 \end{equation}
 The kernel $H(x,y)$ is called the \emph{Hermite kernel}. Integral operators associated with kernels of the form \eqref{hermitkernel} are known as \emph{completely integrable}  \cite{1990IJMPB...4.1003I}.

This representation is much better suited for large-$N$ analysis. For example we could already here apply the Plancherel–Rotach asymptotic formulas to the Hermite polynomials appearing in \eqref{christoffeldarbouxkernel}. We will not use this, however, but rather take a different approach that will allow us to systematically account for subleading corrections.

To this end, we define quantum-mechanical operators
$\bm{H}$ and $\opchi$,
where $\bm{H}$ is an integral operator with the kernel $H(x,y)$:
\begin{equation}
    (\bm{H} f)(x) = \int\limits_{-\infty}^\infty \dd y H(x,y) f(y),
\end{equation}
and $\opchi$ acts by multiplication by the symbol $\chi(x,t)$. The generating functions can then be represented 
as Fredholm determinants:
\begin{equation}
    \langle F^{A,S}(t)\rangle = \det(1+\bm{H}\opchi^{A,S}), \label{fredholmantisym}
\end{equation}
which follows from the definition of the kernel $H(x,y)$ for which we have
\begin{equation}
    \Tr (\bm{H}\opchi)^n= \int\limits_{-\infty}^\infty \dd x_1\,\chi(x_1,t)\cdots
      \int\limits_{-\infty}^\infty \dd x_n\,\chi(x_n,t)\,
      H(x_1,x_2)\cdots H(x_n,x_1) = \Tr M^n.
\end{equation}

The generating function of symmetric/antisymmetric Wilson loops can be represented as a Fredholm determinant of a simpler operator 
at the cost of keeping random-matrix eigenvalues in the symbol, see (3.32), (3.33) in  \cite{Giombi:2020amn}. This construction makes contact to the AdS/CFT integrability via the Quantum Spectral Curve and separation of variables, and it would be interesting to establish such a connection in our case as well\footnote{It can be shown that the two representations, (\ref{fredholmantisym}) in this paper and (3.32), (3.33) in \cite{Giombi:2020amn}, are equivalent (Shota Komatsu, private communication).}.

\subsection{Method of differential equations}
\label{sec:diffeq}

The following approach was developed in \cite{Tracy:1993ah} and recently applied to various observables in supersymmetric theories \cite{Belitsky:2019fan,Belitsky:2020qrm,Bajnok:2024epf}. We follow the notation of the works on the octagon correlator \cite{Belitsky:2019fan,Belitsky:2020qrm}. This work dealt with the Bessel kernel rather than the Hermite kernel, and although the two cases are structurally similar, the analytic properties are quite different and technical steps of the calculation will soon diverge.

The main result of this section is that the Fredholm determinants \eqref{fredholmantisym} can be characterised by a system of integro-differential equations. Importantly, the $N$ dependence of the system of equations is much more amenable to asymptotic analysis. In particular, we will see that the large-$N$ expansion is nothing but the WKB approximation of the resulting equations.

The two representations differ only in the symbol function $\chi^{A,S}$. We first develop the construction for a generic symbol $\chi$ constrained by homogeneity. We will therefore suppress the representation labels $A,S$ throughout this section and restore them when we specialize to the antisymmetric and symmetric representations later on.

The main object of interest is the free energy $\F(t)$ defined in \eqref{eq:free-energy-definition}, which by \eqref{fredholmantisym} can now be expressed in terms of the Fredholm determinant
\begin{equation}
    \F(t) = \frac{1}{N}\log\det(1+\bm{H}\opchi).
    \label{eq:FreeEnergy-FredholmDet}
\end{equation}
Differentiating this expression with respect to $t$, we observe that it is natural to consider a resolvent-type operator:
\begin{equation}
    \bm R
    = (1+\bm H\opchi)^{-1}\bm H.
\end{equation}
Denoting its kernel by $R(x,y)$, the derivative of the free energy can be written as
\begin{equation}
    \partial_t\F(t) = \frac{1}{N}\int \dd x\,R(x,x)\,\partial_t\chi(x,t).
    \label{diffoflogdet}
\end{equation}
This resolvent operator resembles the Hermite kernel twisted by $(1+\bm H\opchi)^{-1}$ and it inherits many of its properties. To use this, we introduce two auxiliary functions \cite{Belitsky:2019fan,Belitsky:2020qrm}:
\begin{equation}
    Q(x,t) = \bra{x}\frac{1}{1+\bm{H}\opchi}\ket{\phi}, \qquad
    P(x,t) = \bra{x}\frac{1}{1+\bm{H}\opchi}\ket{\psi},
\end{equation}
which are the analogues of $\phi(x)$ and $\psi(x)$ defined in \eqref{eq:phi-psi-definition}, similarly twisted by $(1+\bm H\opchi)^{-1}$. The Hermite functions $\phi(x)$ and $\psi(x)$ satisfy a system of differential equations
\begin{equation}
    \phi'(x) = -2Nx\phi(x)+2N\psi(x), \qquad
    \psi'(x) = -2N\phi(x)+2Nx\psi(x), \label{phipsiderivative}
\end{equation}
which follow from the Hermite recurrence relations.

Analogous equations can be derived for the auxiliary functions $Q(x,t)$ and $P(x,t)$. For $Q(x,t)$, this requires computing
\begin{equation}
    \partial_x Q(x,t) = \bra{x}\partial_x(1+\bm{H}\opchi)^{-1}\ket{\phi}.
\end{equation} 
First, one uses an important property of the Hermite kernel
\begin{equation}
    \bra{x}[\partial_x,\bm{H}] \ket{y}=(\partial_x+\partial_y)H(x,y)
    = -2N\phi(x)\psi(y)-2N\psi(x)\phi(y). \label{eq:hermiteDerivativecommutator}
\end{equation}
In addition, to treat the action of the derivative on $\opchi$, we shall require that  symbol functions satisfy the identity
\begin{equation}
    \partial_x\chi(x,t) = \sqrt{\lambda}\,t\partial_t\chi(x,t).\label{symbolfunctionidentity}
\end{equation}
This means that $\chi(x,t)$ is really a function of
\begin{equation}
    z=t e^{\sqrt{\lambda} x},
\end{equation} 
which is indeed how both $\chi^{S}$ and $\chi^{A}$ \eqref{eq:sym-antisym-symbol-functions} depend on $x$ and $t$. A relation like \eqref{symbolfunctionidentity} is necessary to eventually close the system of equations. We could also generalize it to $ \partial_x\chi(x,t) = a(t)\partial_t\chi(x,t)$, which would only slightly modify the contents of this section. However, this would affect the subsequent logic considerably, so we stick to the specific case of $a(t)=\sqrt{\lambda}t$. Note that the expression \eqref{eq:hermiteDerivativecommutator} is different from the similar property of the Bessel kernel. 

Using these relations, the $x$-derivatives of $Q(x,t)$ and $P(x,t)$ can be expressed in terms of $Q$ and $P$. We go through the full derivation in Appendix \ref{app:diffeq}. The resulting equations are
\begin{align}
    &Q'(x,t) = 2N\big[-xQ(x,t)+P(x,t)+2P(x,t)u(t)\big]+\sqrt{\lambda}\,t\partial_t Q(x,t), \label{Qdiffeq}\\
    & P'(x,t) = 2N\big[xP(x,t)-Q(x,t)+2Q(x,t)w(t)\big]+\sqrt{\lambda}\,t\partial_t P(x,t). \label{Pdiffeq}
\end{align}
The equations involve extra functions that are commonly referred to as potentials:
\begin{equation}\label{eq:potentials}
\begin{split}
     u(t) = \bra{\phi}\opchi\frac{1}{1+\bm{H}\opchi}\ket{\phi}&,\quad  w(t) = \bra{\psi}\opchi\frac{1}{1+\bm{H}\opchi}\ket{\psi} \\
     v(t) = \bra{\phi}\opchi &\frac{1}{1+\bm{H}\opchi}\ket{\psi}.
\end{split}
\end{equation}
To close the system of equations, we differentiate the potentials with respect to $t$ and find:
\begin{equation}
 \begin{split}
        \partial_t u(t) &= \int\limits_{-\infty}^\infty \dd x\,Q^2(x,t)\,\partial_t\chi(x,t), \\ 
    \partial_t w(t)&= \int\limits_{-\infty}^\infty \dd x\,P^2(x,t)\,\partial_t\chi(x,t),
    \\
     \partial_t v(t) &= \int\limits_{-\infty}^\infty \dd x\,P(x,t)Q(x,t)\,\partial_t\chi(x,t).
 \end{split} \label{uwdiffeq}
\end{equation}

Two additional, and quite important identities can be derived by manipulating the operators at hand as done in Appendix \ref{app:diffeq}. First, the potential $v(t)$ is directly related to the free energy via
\begin{equation}
    v(t) = \frac{1}{4}\sqrt{\lambda}\,t\partial_t\F(t).\label{eq:vFrelation}
\end{equation}
Furthermore, one of the potentials can be expressed in terms of the other due to the relation:
\begin{equation}
    2u(t) w(t)  - u(t) + w(t) = -\frac{1}{2N}\sqrt{\lambda}\,t\partial_t v(t).
    \label{eq:uwvrelation}
\end{equation}
Since $v(t)$ is directly related to the free energy, we use \eqref{eq:uwvrelation} to eliminate $w(t)$ and write the equation for $P(x,t)$ \eqref{Pdiffeq} in terms of $u(t)$ and $v(t)$. It is also convenient to introduce the differential operator 
\begin{equation}
    \D = \partial_x - \sqrt{\lambda}\,t\partial_t.
\end{equation}
Then the equations \eqref{Qdiffeq}--\eqref{uwdiffeq} combine into the closed system:
\begin{align}
\begin{split}
    \frac{1}{2N}\D Q(x,t) &= -xQ(x,t)+(1+2u)P(x,t)\\
    \frac{1}{2N}\D P(x,t) &= xP(x,t)-\left( \dfrac{1+\frac{1}{N} \sqrt{\lambda}t \partial_t v}{1+2u(t)} \right)Q(x,t)\\
    \partial_t u(t)&= \int\limits_{-\infty}^\infty \dd x\,Q^2(x,t)\partial_t\chi(x,t)\\
    \partial_t v(t)&= \int\limits_{-\infty}^\infty \dd x P(x,t)Q(x,t)\,\partial_t\chi(x,t).
    \label{alldiffeqs}
\end{split}
\end{align}
The closed system of integro-differential equations \eqref{alldiffeqs} is the main result of this section and provides the starting point for the large-$N$ expansion of the free energy, which we consider in the next section. 

This derivation used only properties of the Hermite kernel and that $\chi(x,t)$ satisfies the homogeneity condition \eqref{symbolfunctionidentity}. Since both $\chi^A$ and $\chi^S$ satisfy this relation, the system of equations \eqref{alldiffeqs} holds for both the symmetric and antisymmetric representations, with the two cases distinguished by the choice of the symbol function.

If we assume that $\chi(x,t)$ vanishes as $t\rightarrow 0$, the definitions of the auxiliary potentials imply that $u,v \rightarrow0$ in this limit. The auxiliary functions then reduce to the Hermite functions $Q \rightarrow \phi,P\rightarrow \psi$. Accordingly, we see that the system \eqref{alldiffeqs} in this limit reduces to the system of equations for Hermite functions  \eqref{phipsiderivative}.

\section{WKB expansion}
\label{sec:wkb}
The explicit $N$ dependence of the first two equations for $Q(x,t)$ and $P(x,t)$ in \eqref{alldiffeqs} suggests that the large-$N$ expansion is equivalent to a WKB expansion. We will show that the other two equations for the potentials are compatible with the large-$N$ power counting. The system can therefore be solved order by order in $1/N$, allowing us to determine the free energy through \eqref{eq:vFrelation}. The starting point is the second-order equation for $Q(x,t)$ obtained by eliminating $P(x,t)$ from \eqref{alldiffeqs}. We start by treating the well-familiar case of the harmonic oscillator, to which our equations reduce once $t$ is set to zero.

\subsection{Warm-up: Hermite polynomials}
\label{sec:warmup}
As a starting point, we set $t=0$ which means all potentials vanish and the system \eqref{alldiffeqs} simplifies to
\begin{equation}
    \frac{1}{2N}Q'(x) = -xQ(x)+P(x), \qquad
    \frac{1}{2N}P'(x) = xP(x)-Q(x). \label{firstorderhermite}
\end{equation}
Eliminating $P$ yields the Schr\"{o}dinger equation for the $N$-th harmonic
oscillator level with mass $m=2$ and Planck constant $1/N$:
\begin{equation}
    \frac{1}{4N^2}Q''(x) - x^2Q(x) + \left(1+\frac{1}{2N}\right)Q(x) = 0.
    \label{eq:HermiteQWKB}
\end{equation}
At large $N$ the solution $Q$ assumes a WKB form:
\begin{equation}
    Q(x) = \exp \left(\int^x \dd x' L(x')\right)
\end{equation}
with:
\begin{equation}
    L(x)=N L_{-1}(x)+L_0(x)+\dfrac{1}{N}L_1(x)+\ldots.
\end{equation}
Substituting this ansatz into \eqref{eq:HermiteQWKB} gives a Riccati equation for $L(x)$
\begin{equation}
    L'(x)+L^2(x) = 4N^2 \left(x^2-1-\frac{1}{2N}\right), \label{eq:Ricattit=0}
\end{equation}
which can be solved order by order. After expanding this equation in $1\over N$ and solving for the coefficients, we find that the first two terms are
\begin{equation}
    L_{-1}^2(x)=4(x^2-1)
\end{equation}
\begin{equation}\label{eq:L0t0}
    L_0(x)= -\dfrac{L'_{-1}(x)}{2L_{-1}(x)} - \dfrac{1}{L_{-1}(x)}.
\end{equation}
The second term in the order $N^0$ contribution \eqref{eq:L0t0} 
appears because $Q$ is a highly excited state, with the level quantum number scaling as $n \sim \hbar^{-1}$. 
Higher order terms satisfy the recursion
\begin{equation}
    L_{n+1}(x) = -\dfrac{1}{2 L_{-1}(x)} \left( \sum_{k=0}^{n} L_{k}(x)L_{n-k}(x)+L_{n}'(x) \right) \, , \, n\geq 0.
\end{equation}

The leading-order solution requires taking the square root and consequently has two branches reflecting the existence of two independent solutions of the Schr\"{o}dinger equation. Only one of those is compatible with the normalizability condition: 
\begin{equation}\label{QQ-f}
        \QQ(x)=q_0 e^{S(x)},
\end{equation}
where $q_0$ is a normalization constant to be determined. This mandates us to take
\begin{equation}
    S(x)= -2N\int^x \dd x'\sqrt{x'^2-1} +\ldots 
\end{equation}
and hence choose
\begin{equation}\label{oscillator-1}
    L_{-1}(x)=-2\sqrt{x^2-1}.
\end{equation}
The square root here should be understood in the analytic sense, as a branch of the analytic function defined on the complex plane with a cut from $-1$ to $1$. The analytic square root behaves as $\sqrt{x^2-1}\stackrel{x\rightarrow \pm\infty}{\rightarrow}\pm x$, which guarantees the requisite fall-off at infinity.

The semiclassical wavefunction so constructed is analytic on the complex plane with a cut, but in general is not single-valued unless its phase rotates by an integer multiple of $2\pi$ upon 
encircling the cut:
\begin{equation}
    2N\oint \dd x\,\sqrt{x^2-1}=2\pi i\times{\rm integer}.
\end{equation}
This is exactly the Bohr-Sommerfeld quantization condition,  in our case satisfied by default because the eigenvalue of the Schr\"odinger equation was adjusted from the very beginning to lie at a quantized energy level of the oscillator.

The true solution of the Schr\"odinger equation is a real function of the real variable, while $\QQ(x)$ is complex-valued and is not even defined on the whole real line. A real-valued solution is obtained by taking the symmetric average of its boundary values above and below the cut:
\begin{equation}\label{analytic-c-prescription}
    Q(x)=  \frac{1}{2}\lim_{\epsilon \rightarrow 0}\left(\QQ(x+i\epsilon)+\QQ(x-i\epsilon)\right).
\end{equation}
The resulting wavefunction oscillates in the classically allowed region (the cut) and decays in the forbidden region, as expected. The 
complex-analytic form of the WKB expansion, originally due to Zwaan \cite{zwaan1929}, avoids a cumbersome local analysis near the turning points while automatically reproducing the correct Maslov indices, see \cite{landau1991quantum} for a pedagogical introduction.

The constant $q_0$ is fixed by the normalization condition. Indeed, at
$t = 0$ the function $Q(x)$ coincides with $\phi(x)$, a 
wavefunction of the harmonic oscillator normalized to $1/\sqrt{2}$, see footnote \ref{footnote4}. It is instructive to recover the normalization constant from the semiclassical wavefunction. Simple as it is, this calculation will illustrate the salient features of the subsequent analysis of the full problem.

To calculate the $\mathcal{O}(1)$ normalization we need to know the wavefunction with the same precision, namely to the second order in the WKB expansion. From (\ref{eq:L0t0}) and (\ref{oscillator-1}) we find for the subleading WKB correction:
\begin{equation}
    L_0(x)=-\frac{x}{2(x^2-1)}+\frac{1}{2\sqrt{x^2-1}}\,.
\end{equation}
The two terms appearing here are clearly of a different nature.
In the allowed zone, where the wavefunction oscillates, the first term remains real and therefore determines the genuine (local) normalization of the wavefunction, while the second term becomes imaginary and thus constitutes a correction to the phase. We distinguish such terms by calling them even and odd. Mathematically, the difference is in what happens upon crossing the cut:
\begin{align}\label{even-odd}
    &L(x)=L^{\rm even}(x)+L^{\rm odd}(x)
    \nonumber \\
    &L^{\rm even}(x+i\epsilon)=L^{\rm even}(x-i\epsilon)
     \nonumber \\
    &L^{\rm odd}(x+i\epsilon)=-L^{\rm odd}(x-i\epsilon).
\end{align}

One can notice that 
\begin{equation}
    L^{\rm even}_0=-\frac{1}{2}\,\,\frac{d}{dx}\,
    \ln L_{-1}^{\rm odd}.
\end{equation}
This is a reflection of a more general fact: taking the discontinuity of the Riccati equation (\ref{eq:Ricattit=0}) across the cut we get:
\begin{equation}
    \left(L^{\rm odd}\right)'+2L^{\rm even}L^{\rm odd}=0,
\end{equation}
and hence, at any order in the $1/N$ expansion,
\begin{equation}
    L^{\rm even}=-\frac{1}{2}\,\,\frac{d}{dx}\,
    \ln L^{\rm odd}.
    \label{eq:even-odd-t=0}
\end{equation}
This relation implies the following form of the WKB wavefunction:
\begin{equation}
    \QQ(x)=\frac{q_0}{\sqrt{L^{\rm odd}(x)}}\,
    \exp \left(\int^x \dd x' \, L^{\rm odd}(x')\right),   
\end{equation}
which neatly separates the phase from the modulus.

Explicitly, to the requisite accuracy:
\begin{equation}
    Q(x) = \frac{q_0}{\sqrt{2N}}\frac{1}{(1-x^2)^{1/4}}
    \cos \left(Np(x)-\dfrac{1}{2}\arccos x+\frac{\pi}{4}\right), \qquad
    p(x) = x\sqrt{1-x^2}-\arccos x. \label{eq:phiLargeN}
\end{equation}
This expression agrees with the Plancherel-Rotach asymptotics of the Hermite polynomials for $\phi(x)$ in the oscillatory region $|x|<1$. Outside that region the wavefunction rapidly (exponentially) decays. In computing the normalization we can thus restrict the integral to the allowed zone:
\begin{equation}
    \int \dd x\,Q^2(x)\simeq\frac{q_0^2}{2N}\int\limits_{-1}^{1}\frac{\dd x}{\sqrt{1-x^2}}\,\times\,\frac{1}{2}
    =\frac{q_0^2\pi}{4N}\,,
\end{equation}
where in the last step we replaced the rapidly oscillating $\cos^2$ by $1/2$, again with an exponential accuracy. This fixes the normalization constant to
\begin{equation}
    q_0 = \sqrt{\frac{2N}{\pi}}. \label{eq:normalization}
\end{equation}

\subsection{General \texorpdfstring{$t\neq 0$}{t ≠ 0} case}
\label{sec:wkbt}
We now consider the WKB expansion of the full system of differential equations \eqref{alldiffeqs} for non-zero $t$. For convenience, we define the following variables:
\begin{equation}
    \omega(t) = \frac{1}{2}\sqrt{\lambda}\,t\partial_t\log(1+2u(t)), \qquad
    \nu(t) = \sqrt{\lambda}\,t\partial_t v(t). \label{eq:omeganudef}
\end{equation}
The latter is directly related to the free energy:
\begin{equation}\label{F-from-nu}
    \nu(t)=\frac{\lambda}{4}\,
    \left(t\partial_t\right)^2\F(t),
\end{equation}
as follows from (\ref{eq:vFrelation}).

By eliminating $P(x,t)$ from the system \eqref{alldiffeqs}, we obtain a second-order equation for $Q(x,t)$, which can be compactly written using the notation \eqref{eq:omeganudef} as:
\begin{equation}
    \mathcal{D}^2Q(x,t)+2\omega(t) \mathcal{D}Q(x,t)+Q(x,t) \left[4N^2(1-x^2)+2N\big(1+2\nu(t)+2x\omega(t) \big) \right] = 0.\label{eq:secondorderQeq}
\end{equation}
We choose the WKB ansatz for $Q(x,t)$ in the form
\begin{equation}
    \QQ(x,t) = q_0\sqrt{1+2u(t)}\,e^{S(x,t)}.
\end{equation}
As before, we first determine a complex WKB solution defined on the complex plane with a cut along the allowed zone, and then construct the real WKB wavefunction by symmetrically shifting above and below the cut, as in 
(\ref{analytic-c-prescription}),
\begin{equation}
    Q(x,t)=\dfrac{1}{2}\lim_{\epsilon \rightarrow 0}\left( \QQ(x+i\epsilon,t)+\QQ(x-i\epsilon,t)\right).
\end{equation}
The factor $\sqrt{1+2u(t)}$
is taken out for convenience, it will allow us to write the potential equations \eqref{uwdiffeq} entirely in terms of $\omega(t)$ and $\nu(t)$.

The function $P(x,t)$ enters the first equation in  \eqref{uwdiffeq}
without derivatives and is determined algebraically in terms of $Q(x,t)$:
\newcommand{\PP}{\mathbb{P}}
\begin{equation}
      \PP(x,t) = q_0\frac{e^{S(x,t)}}{\sqrt{1+2u(t)}}
      \left(x+\frac{1}{2N}\big(\D S(x,t)-\omega(t)\big)\right). \label{eq:WKBansatz}
\end{equation}
 We expand the ``semiclassical action'' $S(x,t)$ as
\begin{equation}
   S(x,t)=  N S_{-{1}}(x,t)+\sum_{n=0}^{\infty} N^{-n}S_n(x,t)
\end{equation}
and define
\begin{equation}
    L(x,t)=\mathcal{D}S(x,t) = \sum_{n=-1}^{\infty} N^{-n} L_n(x,t).
\end{equation}
Rewriting the equation \eqref{eq:secondorderQeq} in terms of $L(x,t)$  produces a Riccati equation:
\begin{equation}
    \D L + L^2 + 4N^2 \left(1+\frac{1}{2N}-x^2\right)
    + 4N(\nu+x\omega) + \sqrt{\lambda}\,t\partial_t\omega - \omega^2 = 0.
    \label{eq:ricatti}
\end{equation}
Since the equation now depends on the potentials, in order to solve it order by order we have to assume the structure of their large-$N$ expansion. Therefore, let us assume an ansatz:
\begin{equation}
    \begin{split} &\omega(t)=\sum_{n=0}^{\infty}N^{-n} \omega_n(t), \quad \nu(t)=\sum_{n=0}^{\infty}N^{-n} \nu_n(t),
    \end{split}
    \label{eq:potential-expansion-ansatz}
\end{equation}
which we will show to be consistent with the potential equations. 

To justify the ansatz, let us look at the leading order first. From the Riccati equation we find that $L_{-1}(x,t)$ does not depend on $t$ and is given by:
\begin{equation}
    L_{-1}^2(x) = 4(x^2-1).
\end{equation}
The branch of the square root that corresponds to the normalizable solution is
\begin{equation}\label{ultimateL-1}
    L_{-1}(x)=-2\sqrt{x^2-1} ,
\end{equation}
hence:
\begin{equation}
    S_{-1}(x)=2\int\limits_x \dd x'\sqrt{x'^2-1} 
\end{equation}
Subsequent orders in the WKB expansion can be determined recursively from the Riccati equation
\eqref{eq:ricatti}:
\begin{equation}\label{first-step-recursion}
    2L_{-1}(x)L_0(x,t) + \D L_{-1}+2 
   + 4\big(x\omega_0(t)+\nu_0(t)\big)=0
\end{equation}
  \begin{align}
    L_{n+1}(x,t) =&- \frac{1}{2L_{-1}(x)}\Bigg( \D L_n(x,t)+ \sum_{k=0}^nL_k(x,t)L_{n-k}(x,t) + \sqrt{\lambda}t\partial_t \omega_n(t) 
    \nonumber 
    \\
    &+ 4\big(x\omega_{n+1}(t)+\nu_{n+1}(t)\big) - \sum_{k=0}^n\omega_k(t) \omega_{n-k}(t)\Bigg),\quad n\geq0.
    \label{eq:general-recursion}
\end{align}

The potentials are computed through the matrix elements of $Q$ and $P$. We use the same strategy in evaluating them as in computing the normalization for the semiclassical wavefunction in the warm-up example:  outside the allowed zone the wavefunction decays very fast and up to exponentially small corrections the integration domain can be restricted to the interval $(-1,1)$. There the wavefunction rapidly oscillates, and the oscillatory term $\cos^2(Np(x))$ can be replaced by its average of $1/2$, again with an exponential accuracy. At the end, the matrix elements are only sensitive to the slowly-varying overall amplitude.

The amplitude can be neatly separated from the phase by splitting the ``classical momentum'' $L(x,t)$ into the even and odd parts according to (\ref{even-odd}).
The even part can be integrated using the relation
\begin{equation}
    L^{\rm even}(x,t) 
    =-\frac{1}{2}\,\mathcal{D}\ln \Lo(x,t)
\end{equation}
which is a direct counterpart of \eqref{eq:even-odd-t=0}.

Within the large-$N$ expansion
the Riccati equation is solved algebraically for $L(x,t)$ order-by-order in $1/N$, but to reconstruct the action $S(x,t)$ we need to integrate the relation $\mathcal{D}S(x,t)=L(x,t)$, which now is a partial differential equation in two variables. We can integrate it along the characteristic $z=te^{\sqrt{\lambda}x}$:
\begin{equation}
      S(x,t)=  -{\int\limits_x^{A} \dd y L\left(y,te^{\sqrt{\lambda}(x-y)}\right)},
\end{equation}
which fixes the solution up to an arbitrary function of $z$. To avoid this functional ambiguity we can choose to integrate up to $A=+\infty$ where the boundary conditions on $S$ require normalizability of $\QQ$. The semiclassical wavefunction is then
\begin{equation}\label{eq:QR}   \begin{split}
     \QQ(x,t) = q_0\dfrac{\sqrt{1+2u(t)}}{\sqrt{\Lo(x,t)}} \exp\left(- \int\limits_{x}^A \dd y\Lo\left(y,te^{\sqrt{\lambda}(x-y)}\right)   \right)
\end{split}
\end{equation}
We would like the solution to be normalizable on the whole real axis. In particular, in the region $x>1$ the endpoint integral involves potentially growing contributions proportional to $te^{\sqrt{\lambda}x}$. As we will see at both leading and subleading order, the $t$-dependent terms in $L(x,t)$ have exponentially growing behavior in $x$ and hence decaying in $y \rightarrow \infty$ due to the factors $e^{-\sqrt{\lambda}y}$. The requisite decay at infinity is then guaranteed if we set $A\rightarrow \infty$. There is a caveat, however, related to the $t$-independent part of the momentum. The $t$-independent terms, for instance $L_{-1}(x)$ from (\ref{ultimateL-1}), cannot be integrated to $y\to \infty$. However, these terms will only contribute to the overall normalization of the potentials. The normalization can be fixed independently by matching at $t=0$. Therefore we can subtract these terms and work with:
\begin{equation}\label{eq:DeltaL}
    \hat L (x,t)  = L(x,t)-L(x,0).
\end{equation}
The remaining $t$-dependent integrand decays exponentially as $y\to\infty$ and we can safely take the upper limit $A \rightarrow \infty$. Below we will always assume that the integrand is the subtracted \eqref{eq:DeltaL}, after which we can set $A=\infty$. It was crucial for this argument that the free variable $z=te^{\sqrt{\lambda}x}$, the dependence on which is not fixed by the characteristic equation, grows at $x\rightarrow +\infty$. The boundary condition at $x\rightarrow-\infty$ comes out automatically satisfied because there $z$ decreases.

Integration along characteristics leaves a subtle imprint on the behavior of the wavefunction in the allowed region. Writing
\begin{equation}\label{eq:QR}   \begin{split}
     \QQ(x,t) = q_0\dfrac{\sqrt{1+2u(t)}}{\sqrt{\Lo(x,t)}} \exp\left(- \int\limits_{x}^1 \dd y\Lo\left(y,te^{\sqrt{\lambda}(x-y)}\right)-
     \int\limits_{1}^{+\infty} \dd y\hat{L}^{\rm odd}\left(y,te^{\sqrt{\lambda}(x-y)}\right)   \right),
\end{split}
\end{equation}
we see that the first term in the exponent becomes imaginary between the turning points, but the second term remains real.
This term, which is constant in the ordinary WKB, now gains a dependence on $x$ through its second argument. We introduce a special notation for this contribution:
\begin{equation}
    G(z)=2\int\limits_{1}^{+\infty} \dd y\hat{L}^{\rm odd}\left(y,ze^{-\sqrt{\lambda}y}\right) ,
    \qquad
    z=te^{\sqrt{\lambda}x}.
\end{equation}
The real wavefunction in the allowed region then has the form:
\begin{equation}
    Q(x,t)=q_0\dfrac{\sqrt{1+2u(t)}}{\sqrt{\left|\Lo(x,t)\right|}} \,e^{-\frac{G(z)}{2}}
     \cos(\ldots)
\end{equation}
We do not write the argument of the cosine explicitly, under an integral representing a matrix element $\cos^2$ can be always replaced by $1/2$ with an exponential accuracy. The normalization constant $q_0$ does not depend on $t$ and therefore is the same as in \eqref{eq:normalization}.

Once the wavefunctions $P$ and $Q$ are known, the system of equations (\ref{alldiffeqs}) boils down to two conditions that determine the potentials: 
\begin{align}
    \omega(t) &= \dfrac{\sqrt{\lambda} t N}{\pi}
    \int\limits_{-1}^1 \frac{\dd x}{\left|\Lo(x,t)\right|}
      \,e^{-G(z)}\partial_t\chi(x,t), \label{eq:omegaintegral}\\
    \nu(t) &= \dfrac{\sqrt{\lambda} t N}{\pi}\int\limits_{-1}^1 \frac{\dd x}{\left|\Lo(x,t)\right|}  \,e^{-G(z)}
      \left[x+\frac{1}{2N}\big(\Le(x,t)-\omega(t)\big)\right]\partial_t\chi(x,t),
      \label{eq:nuintegral}
\end{align}
where we took into account the redefinitions in (\ref{eq:omeganudef}). 
The factor of $N$ from the normalization compensates the leading-order behavior of $\Lo$, which starts with $\mathcal{O}(N)$. 
The potentials $\omega(t)$ and $\nu(t)$ are naturally of order one, in agreement with \eqref{eq:potential-expansion-ansatz}.

Explicit iterations of (\ref{first-step-recursion}), 
(\ref{eq:general-recursion}) give:
 \begin{equation}
         L_0^{\rm even}(x,t)= -\frac12 \frac{x}{x^2-1}, \quad L_0^{\rm odd}(x,t) =  \frac12 \frac{1}{\sqrt{x^2-1}}+ \frac{x\omega_0(t)+\nu_0(t)}{\sqrt{x^2-1}},
    \end{equation}
  \begin{equation}
        \begin{split}
               L_1^{\rm even}(x,t) &=-\frac14 \frac{x}{(x^2-1)^2} - \frac12\frac{x\big(x\omega_0(t)+\nu_0(t)\big)}{(x^2-1)^2}- \frac{\sqrt{\lambda}t}{4}\frac{x\omega'_0(t)+\nu'_0(t)}{x^2-1}+\frac14\frac{\omega_0(t)}{x^2-1}\\
               L_1^{\rm odd}(x,t) &= \frac{1}{16}\frac{4x^2+1}{(x^2-1)^{5/2}} +\frac14\frac{\big(x\omega_0(t)+\nu_0(t)\big)^2}{(x^2-1)^{3/2}}+\frac14 \frac{x\omega_0(t)+\nu_0(t)}{(x^2-1)^{3/2}}+\\
         &\frac{\sqrt{\lambda}t}{4}\frac{\omega'_0(t)}{\sqrt{x^2-1}}-\frac14 \frac{\omega^2_0(t)}{\sqrt{x^2-1}}+\frac{x\omega_1(t)+\nu_1(t)}{\sqrt{x^2-1}}.
         \label{eq:L1-recursion}
        \end{split}
    \end{equation}

\subsection{Leading order}
\label{sec:leadingorderwkb}
We will now compute the leading-order potentials by solving the integral equations at $\mathcal{O}(N^0)$. At this order, 
\begin{equation}
    \left|\Lo(x,t)\right|=2N\sqrt{1-x^2}\,,
\end{equation}
and
\begin{equation}
    G_0(z) = 2\int\limits_1^\infty \dd y\,
    \frac{y\,\omega_0(ze^{-\sqrt{\lambda}y})+\nu_0(ze^{-\sqrt{\lambda}y})}{\sqrt{y^2-1}},
    \label{eq:G0def}
\end{equation}
which depends on $z=t e^{\sqrt{\lambda}x}$. It will be useful to denote the Euler derivative of a function with respect to its argument by $\theta$
\begin{equation}
    \theta f(\zeta) = \zeta \partial_\zeta f(\zeta).
\end{equation}

The integral equations (\ref{eq:omegaintegral}),
(\ref{eq:nuintegral}) at the leading order become
\begin{equation}
    \omega_0(t) = \frac{\sqrt{\lambda}}{2\pi}\int\limits_{-1}^1\frac{\dd x}{\sqrt{1-x^2}}e^{-G_0(z)}\theta\chi(z), \qquad
    \nu_0(t) = \frac{\sqrt{\lambda}}{2\pi}\int\limits_{-1}^1\frac{x\,\dd x}{\sqrt{1-x^2}}e^{-G_0(z)}\theta \chi(z).\label{eq:leadingorder-potential}
\end{equation}
Here we write $\chi(x,t) = \chi(z)$ since \eqref{symbolfunctionidentity} implies that the symbol depends on $x$ and $t$ only through $z$, and thereby $t\partial_t\chi$ can be replaced by $z\partial_z\chi$.

These equations can be reformulated as a single equation on $G_0(z)$, by substituting \eqref{eq:leadingorder-potential} into the definition \eqref{eq:G0def}:
\begin{equation}\label{eq:G0integral}
  G_0(z) =  \frac{\sqrt{\lambda}}{\pi} \int_{1}^{\infty} \dd y \int_{-1}^{1} \dd s   \dfrac{s+y}{\sqrt{1-s^2}\sqrt{y^2-1}}  e^{-G_0(ze^{\sqrt{\lambda}(s-y)})}\theta \chi(ze^{\sqrt{\lambda}(s-y)})\equiv\frac{1}{\pi}\mathsf{T}[e^{-G_0}\theta\chi](z).
\end{equation}
The second equality defines a linear integral operator $\mathsf{T}$ through a fairly complicated double integral, which can be regarded as a generalized Abel transform. In addition, the argument is non-linear in the unknown function $G_0$  that we want to find. The problem looks hopelessly complicated. It is nonetheless possible to proceed because of a remarkable identity satisfied by the operator $\mathsf{T}$. In fact $\mathsf{T}$ is just $\theta^{-1}$ in disguise:
\begin{equation}
   z\dfrac{\dd}{\dd z} \mathsf{T}[F](z) = \pi F(z). \label{integraltransformidentity}
\end{equation}
We established this magic identity at first by experimentally checking it for a number of probe functions, but the general proof is not very difficult and is presented in the Appendix \ref{sectionIntegralTransform}. 

Differentiating both sides of \eqref{eq:G0integral} and using \eqref{integraltransformidentity} we obtain a simple differential equation for $G_0$:
\begin{equation}
    G_0'(z) = \chi'(z)\,e^{-G_0(z)}.
    \label{eq:G0-prime}
\end{equation}
The magic identity for the integral transform results in a differential equation for $G_0(z)$ that can be easily solved. Imposing the boundary condition $G_0(0)=0$ and assuming that $\chi(0) = 0$ (which is true for symbols of the symmetric and anti-symmetric generating functions), we find:
\begin{equation}
    G_0(z)=\log(1+\chi(z)). 
\end{equation}

With $G_0(z)$ at hand, we can proceed to evaluate $\nu_0(t)$ which is related to $\F_0(t)$ through (\ref{F-from-nu}):
\begin{equation}
    \nu_0(t) = \frac{\lambda}{4} (t\partial_t)^2\F_0(t) \label{eq:nu0F0relation}.
\end{equation}
In particular integrating the second equation in 
(\ref{eq:leadingorder-potential}) by parts 
and using $\partial_x f(z) = \sqrt{\lambda}z\partial_z f(z)$, we find: 
\begin{align}
        \nu_0(t) &= \frac{\lambda}{2\pi} (t\partial_t)^2 \int_{-1}^1 \dd x \sqrt{1-x^2} \log\big(1+\chi(z)\big).\label{eq:nu0simplified}
\end{align}
Comparing \eqref{eq:nu0F0relation} and \eqref{eq:nu0simplified}, we obtain the leading-order free energy in the form of a simple convolution integral:
\begin{equation}
    \F_0(t) = \frac{2}{\pi}\int_{-1}^1\dd x\,\sqrt{1-x^2}\,\log\big(1+\chi(z)\big).
    \label{eq:F0general}
\end{equation}

This result is quite general and is valid for any symbol function.
Specifically for the generating functions of antisymmetric and symmetric Wilson loops, the symbols are
\begin{equation}
    \chi^{A}(z)=z \,,\quad \chi^{S}(z) =\dfrac{z}{1-z},
\end{equation}
giving
\begin{align}
    &\F^A_0(t) = \frac{2}{\pi}\int_{-1}^1\dd x\,\sqrt{1-x^2}\,
    \log\left(1+te^{\sqrt{\lambda}x}\right), \label{F0anti}\\
    & \F^S_0(t) =- \frac{2}{\pi}\int_{-1}^1\dd x\,\sqrt{1-x^2}\,
    \log\left(1-te^{\sqrt{\lambda}x}\right). \label{F0sym}
\end{align}
These formulas reproduce the large-$N$ matrix-model 
result of \cite{Hartnoll:2006is}.

\subsection{Next-to-leading order}
We now consider the first subleading correction in the WKB expansion. It will be useful to define
\begin{align}\label{eq:notations}
    &n(z) = \theta G_0(z) = z\partial_z  \log(1+\chi(z)),\\
    &\Omega_i(x,t) = x\omega_i(t) + \nu_i(t),\\
    &\tau_y(z)= z e^{-\sqrt{\lambda}y},\\
    & H_i(z)
  =
  2\int\limits_1^\infty
\frac{\dd y}{\sqrt{y^2-1}}\,
\Omega_i \big(y,\tau_y(z)\big).
\label{eq:G1-definition}
\end{align}

To obtain potentials at the next order, we note that the measure in (\ref{eq:omegaintegral}), (\ref{eq:nuintegral}) receives two types of corrections: from $G_1$ in the exponential and from $\Lo_0$ in the pre-factor:
\begin{equation}
    e^{-G}=e^{-G_0}\left(1-\frac{1}{N}\,G_1+\ldots\right),
\end{equation}
and
\begin{equation}
    \frac{1}{|L^{\rm odd}|} = \frac{1}{N |L_{-1}^{\rm odd}|} 
  \left(
  1-\frac{1}{N}
  \frac{L_0^{\rm odd}}{L_{-1}^{\rm odd}}
  +\ldots
  \right)
  =\frac{1}{2N\sqrt{1-x^2} } 
  \left(
  1-\frac{
  1+2\Omega_0
  }{
  4N(1-x^2)
  }+\ldots
  \right)
  \label{eq:one-overLodd-expansion}
\end{equation}

The subleading correction $G_1$ is obtained by integrating the $t$-dependent part of $L^{\rm odd}_1$ from \eqref{eq:L1-recursion} (that is, the whole expression with the first term dropped).  The action naturally decomposes into a contribution involving the subleading potentials and a remainder built from the leading-order data already known:
\begin{equation}
    G_1(z) = H_1(z) + R_1(z),
    \label{eq:S1-equation}
\end{equation}
with the remainder given by
\begin{align}
  R_1(z)
  ={}&
 \frac12 \int\limits_1^\infty\dd y\,
\Bigg[
\frac{
\Omega_0 \big(y,\tau_y(z)\big)^2
+\Omega_0 \big(y,\tau_y(z)\big)
}{
(y^2-1)^{3/2}
}
+\frac{
\sqrt{\lambda}\,
\theta\omega_0 \big(\tau_y(z)\big)
-\omega_0 \big(\tau_y(z)\big)^2
}{
\sqrt{y^2-1}
}\Bigg].
  \label{eq:R1-before-finite-part}
  \end{align}
After some further simplification, the equations for the first subleading  potentials boil down to
\begin{align}
    &\omega_1(t) =- \frac{\sqrt{\lambda}}{2\pi} \int\limits_{-1}^1  \frac{\dd x}{\sqrt{1-x^2}}\,
  n \left(t e^{\sqrt{\lambda}x}\right)
  \Bigg[
  G_1 \left(t e^{\sqrt{\lambda}x}\right)
  +\frac{
  1+2\Omega_0(x,t)
  }{
  4(1-x^2)
  }
  \Bigg]\\
  &\nu_1(t) =- \frac{\sqrt{\lambda}}{2\pi}\int\limits_{-1}^1
  \frac{\dd x}{\sqrt{1-x^2}}\,
  n \left(t e^{\sqrt{\lambda}x}\right)
  \Bigg[
  xG_1 \left(t e^{\sqrt{\lambda}x}\right)
  +\frac{
  \omega_0(t)+x\nu_0(t)
  }{
  2(1-x^2)
  }
  \Bigg]
\end{align}
Inserting expressions for $\omega_1$ and $\nu_1$ directly into the definition of $H_1(z)$ \eqref{eq:G1-definition} gives
  \begin{equation}
  H_1(z)
  =-
  \frac{\sqrt{\lambda}}{\pi}
  \int\limits_1^\infty
  \frac{\dd y}{\sqrt{y^2-1}}\,
   \int\limits_{-1}^{1}
  \frac{\dd s}{\sqrt{1-s^2}}\,
  n \big(\tau_{y-s}(z)\big)
  (y+s)G_1 \big(\tau_{y-s}(z)\big)
  +R_2(z).
  \label{eq:G1-after-substitution}
  \end{equation}
As before we separate a genuine subleading contribution from the remainder inherited from the previous order:
\begin{align}
  R_2(z)
  &=-
  \frac{\sqrt{\lambda}}{\pi}
  \int\limits_1^\infty
  \frac{\dd y}{\sqrt{y^2-1}}\,
   \int\limits_{-1}^{1}
  \frac{\dd s}{\sqrt{1-s^2}}\,
  n \big(\tau_{y-s}(z)\big)
  \nonumber\\
  &\times
  \Bigg[
  \frac{y}{4(1-s^2)}
  +\frac{(y+s)
 \Omega_0\big(s,\tau_y(z)\big)}{2(1-s^2)}
  +\frac12\omega_0 \big(\tau_y(z)\big)
  \Bigg].
\label{eq:R2-before-moments}
  \end{align}
Quite remarkably, the remainder here is just the same as in 
(\ref{eq:S1-equation}), but with the opposite sign. This is not immediately obvious, as the integrals in (\ref{eq:R1-before-finite-part}),
(\ref{eq:R2-before-moments}) look quite different, but we can show
nonetheless that
\begin{equation}\label{eq:R1R2}
    R_1(z)+R_2(z) =0.
\end{equation}
The derivation can be found in the Appendix~ \ref{app:WKBdetails}.

As for the first term in $H_1$, it is precisely the integral transform defined in (\ref{eq:G0integral}). We thus obtain
  \begin{equation}
      H_1(z) =-\frac{1}{\pi} {\sf T}[n G_1](z) +R_2(z).
      \label{eq:G_1(z)}
  \end{equation}
Differentiating both sides and using the magic identity (\ref{integraltransformidentity})
we get:
\begin{equation}\label{eq:G1DE}
    H_1'(z) = -G_0'(z) \Big[H_1(z)+R_1(z)\Big]+R_2'(z).
\end{equation}
Using (\ref{eq:R1R2}) this can be rewritten as
\begin{equation}
    H_1'+G_0'H_1=-R_1'-G_0'R_1.
\end{equation}
This equation can be readily integrated. The most obvious solution 
\begin{equation}
    H_1=-R_1
\end{equation}
is the right one, given
the boundary conditions $H_1(0)=R_1(0)=0$\footnote{Since $R_1=-R_2$ and the defining integrand of $R_2(z)$ contains $n(\tau_{y-s}(z))$ which vanishes at $z=0$, we conclude that $R_1(0)=0$.}. This, in turn, implies that
\begin{equation}
    G_1 = 0.
\end{equation}

The first subleading potentials therefore reduce to 
  \begin{align}
  \omega_1(t)
  &=
  -\frac{\sqrt{\lambda}}{8\pi}\,
  \int\limits_{-1}^{1}
  \frac{\dd x}{(1-x^2)^{3/2}}\,
  n \left(t e^{\sqrt{\lambda}x}\right)
  \Big[
  1+2x \omega_0(t)+2\nu_0(t)
  \Big],
  \label{eq:omega1-local-start}
  \\
  \nu_1(t)
  &=
  -\frac{\sqrt{\lambda}}{4\pi}\,
  \int\limits_{-1}^{1}
  \frac{\dd x}{(1-x^2)^{3/2}}\,
  n \left(t e^{\sqrt{\lambda}x}\right)
  \Big[
  \omega_0(t)+x\nu_0(t)
  \Big],
  \label{eq:nu1-local-start}
  \end{align}
and are thereby expressed through the known leading-order data.

The integrals here and actually in many expressions encountered in this section are formally divergent, as $(1-x^2)^{-{\nu \over 2}}$ integrated between $-1$ and $1$ or $(y^2-1)^{-{\nu \over 2}}$ integrated from $1$ to infinity. Such integrals should be understood in the analytic sense. Indeed, the wavefunction was originally defined in the complex plane and if all the $\epsilon$'s were kept throughout the calculation the integrals would run along the contours surrounding the cut. Such integrals are well defined. A similar phenomenon  arises in the textbook perturbation theory of quantum mechanics.
  
The integrals can be readily evaluated with this contour-integral prescription,  since $\omega_0$ and $\nu_0$ depend only on $t$, and the integrals with $n \left(t e^{\sqrt{\lambda}x}\right)$ reduce to the leading-order potentials, which are collected in the identities \eqref{eq:fp-moment-zero} and \eqref{eq:fp-moment-one}. Putting everything together we get:
\begin{align}
      &\omega_1(t)
  =
  \frac{\sqrt{\lambda}}{4}\,
t\partial_t
  \left[
  \omega_0(t)^2+\nu_0(t)^2+\nu_0(t)
  \right],
  \label{eq:omega1-final}\\
  &\nu_1(t)=\frac{\sqrt{\lambda}}{2}\,
 t\partial_t
  \left[
  \omega_0(t)\nu_0(t)
  \right].
  \label{eq:nu1-final}
\end{align}

Using \eqref{eq:vFrelation}, we obtain the subleading correction to the free energy $\F_1(t)$ from $\nu_1(t)$:
\begin{align}
\begin{split}
    t\partial_t \F_1(t) &= \frac{2}{\sqrt{\lambda}}\omega_0(t)\nu_0(t) \\
    &= \frac{\sqrt{\lambda}}{2\pi^2}\int_{-1}^1 \frac{\dd x}{\sqrt{1-x^2}}\int_{-1}^1 \frac{{s\,\dd s}}{\sqrt{1-s^2}}\,n\left(te^{{\sqrt{\lambda}x}}\right) n\left(te^{\sqrt{\lambda}s}\right).
\label{eq:tpartialtF1}
\end{split}
\end{align}
This way of presenting the result is particularly simple and it applies to a general symbol function satisfying $\mathcal{D}\chi=0$. The differential relation \eqref{eq:tpartialtF1} can be integrated producing:
\begin{equation}
    \mathcal{F}_1(t)= \dfrac{1}{8\pi^2}\int_{-1}^{1} \dd x\int_{-1}^{1} \dd s  \left[\dfrac{ \log\left(1+\chi\left(t e^{\sl x}\right) \right)-\log\left(1+\chi\left(t e^{\sl s}\right) \right) }{s-x} \right]^2 \dfrac{1- xs}{\sqrt{1-s^2} \sqrt{1-x^2}}.
    \label{eq:general-F1}
\end{equation}
This formula has been observed at weak coupling \cite{Okuyama:2017feo} for the antisymmetric case $\chi(z) = z$. Here we prove directly that it holds generically for a symbol function $\chi(z)$ and at arbitrary coupling. 

As was already noticed in \cite{Okuyama:2017feo}, the kernel in this formula is the connected  density-density correlator in the Gaussian matrix model
\cite{ambjorn1990multiloop}\footnote{The density-density correlator is actually universal, independent of the potential in the matrix model \cite{ambjorn1990multiloop,Makeenko:1991tb}}. The result can thus be concisely written as 
\begin{equation}
    \mathcal{F}_1(t)=\left\langle
       \ln\det\left(1+\chi\left(t e^{\sl X}\right)\right)
        \ln\det\left(1+\chi\left(t e^{\sl X}\right)\right)
    \right\rangle_{\rm conn}
\end{equation}
This is the most general form of the result. Simpler formulas can be obtained for the symmetric and antisymmetric cases, which can be compared with the existing literature.

For the antisymmetric Wilson loop we have
\begin{equation}
    \chi^A(z) = z, \quad G_0^A(z) = \log(1+z), \quad n^A(z) = \frac{z}{1+z}.
\end{equation}
Integrating \eqref{eq:tpartialtF1} and imposing $\F^A_1(0) = 0$ gives
\begin{align}
    \F_1^A(t)&=  \frac{\sqrt{\lambda}}{2\pi^2}\int_{-1}^1 \frac{\dd x}{\sqrt{1-x^2}}\int_{-1}^1 \frac{{s\, \dd s}}{\sqrt{1-s^2}} \int_0^t \dd \tau \frac{
    \tau e^{\sqrt{\lambda}(x+s)}
    }{
    \left(1+\tau e^{\sqrt{\lambda}x}\right)
    \left(1+\tau e^{\sqrt{\lambda}s}\right)}.
\end{align}
The $\tau$-integral can be evaluated and after some simplifications we find
\begin{equation}
    \F_1^A(t) = -\frac{\sqrt{\lambda}}{2\pi^2}\int_{-1}^1 \frac{\dd x}{\sqrt{1-x^2}}\int_{-1}^1 \frac{{s\, \dd s}}{\sqrt{1-s^2}} \frac{e^{\sqrt{\lambda}x}}{e^{\sqrt{\lambda}s}-e^{\sqrt{\lambda}x}}\log \left(\frac{1+t e^{\sqrt{\lambda}s}}{1+te^{\sqrt{\lambda}x}} \right).
    \label{eq_final-antisymmetric-F1}
\end{equation}
To compare with the previous work \cite{CanazasGaray:2018cpk}, we must account for a difference in conventions. In particular, their free energy is related to ours by $\widetilde{\F}(t) = \F^A(te^{-\lambda/8N})$\footnote{They define $\widetilde{\F}(t)= \frac1N \log \det (1+tI(g,g))$ where $M^A(t) = t e^{\lambda/8N}I(g,g)$.
}, which gives 
\begin{equation}
    \widetilde{\F}_1(t)
    =\F_1^A(t)-\frac{\lambda}{8}\,
    t\partial_t\F_0^A(t).
\end{equation}
Including this shift, our result agrees with
\cite{CanazasGaray:2018cpk}. Moreover, after some rearrangements, \eqref{eq_final-antisymmetric-F1} also agrees with the result of \cite{Gordon:2017dvy}.

The subleading correction to the free energy is consistent with the symmetry:
\begin{equation}
    n^{S}(z)=-n^A(-z).
\end{equation}
Therefore we have
\begin{equation}
    \omega_0^S(t)=-\omega_0^A(-t), \quad \nu_0^S(t)=-\nu_0^A(-t),
\end{equation}
which gives a relation between the free energies
\begin{align}\label{eq:F1fullintegral}
    \F_1^S(t) &= \F_1^A(-t)\nn \\
    &=  -\frac{\sqrt{\lambda}}{2\pi^2}\int_{-1}^1 \frac{\dd x}{\sqrt{1-x^2}}\int_{-1}^1 \frac{{s\, \dd s}}{\sqrt{1-s^2}}\frac{e^{\sqrt{\lambda}x}}{e^{\sqrt{\lambda}s}-e^{\sqrt{\lambda}x}}\log \left(\frac{1-t e^{\sqrt{\lambda}s}}{1-te^{\sqrt{\lambda}x}} \right)
\end{align}
As we discussed previously, this relation was conjectured on account of the weak coupling expansion \cite{Okuyama:2018aij} and studied further in \cite{Fiol:2018yuc,CanazasGaray:2019mgq}. Our calculation independently confirms  this relation for finite $\lambda$ and $t$ to the first two orders in $1/N$.

\section{Saddle-point analysis}
Having determined the generating functions through the first subleading order in $1/N$ for both symmetric and antisymmetric representations, the next step is to compute the Wilson loops at large $N$ and large $k$, with $f=k/N$ fixed, by analyzing the saddle point in the $t$ contour integral. Moreover, we will consider the strong coupling limit, where matrix model results can be directly compared with the D-brane description in the dual theory.

Recall that the $k$-th Wilson loop is obtained from
\begin{equation}
    W_{A_k,S_k} = \oint_C \frac{\dd t}{2\pi i t} \exp\Big( N \big[\F^{A,S}(t)-f\log t\big]\Big).\label{eq:Wilsonloop-extraction}
\end{equation}
Following \cite{Hartnoll:2006is}, it is natural to change variables $t=e^{\sqrt{\lambda}y}$, which maps the complex $t$-plane to a cylinder. In terms of these variables, Eq. \eqref{eq:Wilsonloop-extraction} becomes
\begin{equation}\label{eq:tintegrallargeN}
    W_{A_k,S_k} = \frac{\sqrt{\lambda}}{2\pi i}\oint_C \dd y \exp\Big(N\Phi_0^{A,S}(y) + \F_1^{A,S}(e^{\sqrt{\lambda}y}) + \cO(N^{-1}) \Big),
\end{equation}
where
\begin{equation}
    \Phi_0^{A,S}(y) = \F_0^{A,S}(e^{\sqrt{\lambda}y}) - f\sqrt{\lambda}y.
\end{equation}
The integral is evaluated via an expansion around a saddle point, which is determined by the $\Phi_0^{A,S}$ term:
\begin{equation}
    \frac{\dd \Phi^{A,S}_0}{\dd y} (y_*) = 0.
    \label{eq:GeneralSaddlepoint}
\end{equation}
Substituting the expression $\mathcal{F}^{A,S}_0(t)$ for the symmetric and antisymmetric representations \eqref{F0anti}--\eqref{F0sym} into the saddle point condition we get the equation:
\begin{equation}
 \pm \dfrac{2 \sqrt{\lambda}}{\pi}\int_{-1}^1  \dd x \dfrac{\sqrt{1-x^2}}{1 \pm e^{-\sqrt{\lambda}(x+y_*)}}    -f \sqrt{\lambda} =0, 
 \label{eq:General-explicit-saddlepoint-eq}
\end{equation}
where the upper sign corresponds to the antisymmetric representation and the lower sign to the symmetric representation. It is well-known that the analysis of this equation differs significantly depending on the sign. In fact this is already true for the integral \eqref{eq:tintegrallargeN} itself. The leading-order free energies $\mathcal{F}_0^{A}$ and $\mathcal{F}_0^{S}$ have different analytic structures in the complex $y$-plane. From their explicit integral expressions we see that the function $\mathcal{F}_0^{A}$ is regular on the real $y$-axis, while  $\mathcal{F}_0^{S}$ develops a branch cut at real $y$ because of the vanishing denominator in the logarithm. The same applies to the subleading corrections. Therefore the two cases should be treated separately.

\subsection{Antisymmetric representation}
Let us briefly recall how the computation works for the antisymmetric representation at leading order. The large-$k$ and strong coupling analysis at subleading order was computed thoroughly in \cite{CanazasGaray:2018cpk,Gordon:2017dvy}. Since our finite-coupling integral representation of $\F_1^A$ agrees with these results, we do not repeat the corresponding subleading saddle-point analysis here.

For the antisymmetric representation, we take the $+$ sign in \eqref{eq:General-explicit-saddlepoint-eq}. 
At strong coupling, the factor in the integrand \eqref{eq:General-explicit-saddlepoint-eq} becomes
\begin{equation}
    \frac{1}{1+e^{-\sqrt{\lambda}(x+y_*)}} \to \Theta(x+y_*).
\end{equation}
It is convenient to parametrize the saddle point by
\begin{equation}
    y_*=-\cos \theta_*, \quad 0<\theta_*<\pi.
\end{equation}
Then the saddle-point equation becomes
\begin{align}
   f&= \frac{2}{\pi}\int_{\cos \theta_*}^1 \dd x\sqrt{1-x^2} = \frac{1}{\pi}\left( \theta_* - \cos \theta_* \sin \theta_*\right).
\end{align}
Evaluating the leading action at this point gives
\begin{equation}
    \Phi_0^A(y_*) = \frac{2\sqrt{\lambda}}{3\pi}\sin^3\theta_*.
\end{equation}
The leading contribution to the contour integral \eqref{eq:Wilsonloop-extraction} for large $k\sim N$ is therefore
\begin{equation}
    W_{A_k} = \exp\lb \frac{2N \sqrt{\lambda}}{3\pi}\sin^3\theta_* \rb.
\end{equation}
The subleading term was previously evaluated  \cite{Gordon:2017dvy,CanazasGaray:2018cpk} from the same generating function $\mathcal{F}_1^A(t)$. We simply quote the result below:
\begin{equation}
    W_{A_k} = \exp\lb \frac{2N \sqrt{\lambda}}{3\pi}\sin^3\theta_*+\dfrac{\lambda}{8\pi^2}\left( \sin^2 \theta_*-\theta_* \sin2\theta_* +\theta_*^2 \right) \rb.
\end{equation}
This result is valid in the strict large-$\lambda$ approximation, corrections in $1/\sqrt{\lambda}$ were investigated in \cite{Horikoshi:2016hds}.

\subsection{Symmetric representation}

The symmetric representation has a qualitatively different analytic structure, which requires some care. The leading-order contribution was originally computed in \cite{Hartnoll:2006is}. We first reproduce that computation and then follow the same approach for the subleading order. 

The saddle-point equation can be written as (after changing $x$ to $-x$ as the integration variable):
\begin{equation}\label{eq:symmetricsaddle1}
\dfrac{2 }{\pi}\int_{-1}^1  \dd x \dfrac{\sqrt{1-x^2}}{e^{\sqrt{\lambda}(x-y_*)}-1} -f  =0.
\end{equation}
Its solution depends on the coupling. At weak coupling the saddle point exists at some value $y_{*} <-1$. Increasing $\lambda$ one can observe that the critical point moves towards $y_*=-1$ and reaches this endpoint at some critical coupling $\lambda_c$. At larger values of the coupling the solution does not exist.

One way to understand this behavior is to interpret the integral in \eqref{eq:symmetricsaddle1} as the particle number of a non-interacting Bose gas with chemical potential $y_*+1$ and a single-particle density of states given by the Wigner semicircle. The whole distribution is shifted to start at the lower endpoint of the Wigner semicircle $x=-1$. The filling is determined by $f$, so we are looking for a value of the chemical potential that would realize this filling $f$.

At sufficiently large temperatures $T \sim 1/\sqrt{\lambda}$ and fillings of order $f< T$, this is always possible. As $T$ decreases, the Bose distribution becomes concentrated along the $x=-1$ edge of the semicircle, causing the chemical potential to increase to maintain the same filling until it eventually reaches zero at $y_*=-1$. The critical filling can be computed at weak coupling by expanding the exponential at $y_*=-1$:
\begin{equation}
    \dfrac{2}{\pi}\int_{-1}^1\dd x \dfrac{\sqrt{1-x^2}}{x+1} \sim \sqrt{\lambda}f_c \quad \Rightarrow \quad f_c  \sim \dfrac{2}{\sqrt{\lambda}}.
\end{equation}

The Wilson loop itself, however, does not have a singularity at this point. There are two ways to deal with this issue at the level of the saddle-point computation. The first is to reevaluate the contribution of different eigenvalues. One can show that for the symmetric representation, the Bose distribution exponentially enhances the contribution of the largest eigenvalue, which may lie outside of the Wigner semicircle. The $t$-contour can then be deformed to pick out this contribution.

Here we instead follow the proposal of \cite{Hartnoll:2006is}, which is to recognize that the function
\begin{equation}
    \F_0^S(e^{\sqrt{\lambda}y}) = -\int_{-1}^1 \dd x \rho(x){\rm Log}\Big(1-e^{\sqrt{\lambda}(x+y)}\Big), \quad \rho(x) =  \frac{2}{\pi}\sqrt{1-x^2}
\end{equation}
has a branch cut along $[-1,1]$ in the $y$ plane. Therefore to solve for the saddle point carefully, one has to analytically continue this function to the second sheet. Then, at large coupling, one finds that the saddle-point equation does have a solution that lies on the second sheet of its Riemann surface. These two approaches can be shown to be equivalent at least to leading order in large $N$ \cite{Chen-Lin:2015dfa}. 

Since we will extend the analytic approach beyond leading order, we start by reviewing the computation of \cite{Hartnoll:2006is}. Note that the initial contour in \eqref{eq:Wilsonloop-extraction} in the $y$ plane is the interval $[c-\frac{i\pi}{\sqrt{\lambda}},c+\frac{i\pi}{\sqrt{\lambda}}]$ located at $\Re \, y=c<-1$. One can push the contour towards $\Re \, y \rightarrow \infty$, where the action tends to $\Phi_0(y)\to -\infty$ and the contribution of the integral vanishes. In crossing the cut one picks up a contour integral around it. 

For a function with a cut along this interval, we denote its boundary values above and below the cut by
\begin{equation}
    X_\pm(y) =  \lim_{\epsilon\to0^+}X(y\pm i\epsilon), \quad -1<y<1.
\end{equation}
Then the integral \eqref{eq:Wilsonloop-extraction} can be rewritten in terms of the discontinuity on the cut:
\begin{equation}
    W_{S_k} = -\frac{\sqrt{\lambda}}{\pi}\,\Im \int_{-1}^{1} \dd y \, e^{N\Phi_{0,-}^S(y)+\F_{1,-}^S(e^{\sqrt{\lambda}y})}.
\end{equation}
First we calculate the function below the cut. Approaching the cut from below or above dictates the branch of the logarithm: 
\begin{equation}
    {\rm Log}\Big(1-e^{\sqrt{\lambda}(x+y-i0)}\Big) = \log|1-e^{\sqrt{\lambda}(x+y)}| + i\pi \Theta(x+y).
\end{equation}
Then the leading saddle-point action is 
\begin{equation}
    \Phi^{S}_{0,\pm}(y) = -\int_{-1}^1 \dd x \rho(x) \log|1-e^{\sqrt{\lambda}(x+y)}| \pm 2i \delta(y)-4\kappa y,
\end{equation}
where
\begin{equation}
    \delta(y) = \int_{-1}^y \dd x \sqrt{1-x^2} = \frac12 \left( y\sqrt{1-y^2}+ \arcsin y + \frac{\pi}{2}\right)
\end{equation}
and
\begin{equation}
    \kappa = \frac{k\sqrt{\lambda}}{4N}.
\end{equation}
The discontinuity across the cut is therefore
\begin{equation}
    \Phi_{0,+}^{S}(y) -\Phi_{0,-}^{S}(y) = 4i \delta(y).
    \label{eq:disc-across-cut}
\end{equation}
Now we can continue this function across the cut. Denote the two Riemann sheets by $\rm I$ and $\rm II$. The two sheets are glued across the cut by requiring
\begin{equation}
    \Phi_{0,+}^{S,\mathrm{II}}(y)
    =
    \Phi_{0,-}^{ S,\mathrm I}(y).
    \label{eq:symmetric-sheet-matching}
\end{equation}
Together with \eqref{eq:disc-across-cut}, this defines the analytic continuation away from the cut as
\begin{equation}
    \Phi_{0}^{ S,\mathrm{II}}(y)
    =
    \Phi_0^{ S,\mathrm I}(y)-4i\delta(y).
    \label{eq:symmetric-second-sheet-action}
\end{equation}
As discussed above, at sufficiently weak coupling, there is a saddle point on the real $y$-axis with $y <-1$, i.e. at $\lambda<\lambda_c$ we have a solution to the equation:
\begin{equation}
    \frac{\dd \Phi_0^{ S,\mathrm{I}}}{\dd y}(y_*) = 0.
\end{equation}
For $\lambda>\lambda_c$ we switch to the second sheet and look for:
\begin{equation}
    \frac{\dd \Phi_0^{ S,\mathrm{II}}(y_*)}{\dd y} = 0.
\end{equation}
Explicitly, this equation reads
\begin{equation}
    \sqrt{\lambda}\int_{-1}^1 \dd x \rho(x) \frac{e^{\sqrt{\lambda}(x+y_*)}}{1-e^{\sqrt{\lambda}(x+y_*)}} - 4\kappa - 4i\sqrt{1-y_*^2}=0.
    \label{eq:symmetric-finite-coupling-saddle}
\end{equation}

We now consider the strong coupling limit with $\kappa$ fixed. For $y_*<-1$, the integral in \eqref{eq:symmetric-finite-coupling-saddle} is exponentially suppressed in the strong coupling limit. Solving \eqref{eq:symmetric-finite-coupling-saddle} we find
\begin{equation}
    y_* = -\sqrt{1+\kappa^2} 
\end{equation}
up to exponentially suppressed terms. Evaluating \eqref{eq:symmetric-second-sheet-action} at the saddle point gives the leading large-$N$ contribution to the symmetric Wilson loop
\begin{equation}
    W_{S_k} = \exp\left[2N\lb \kappa \sqrt{1+\kappa^2}+ {\rm arcsinh}\, \kappa \rb\right].
\end{equation}

\paragraph{Subleading order.} 

We need to evaluate the subleading $\mathcal{F}_1^S(e^{\sqrt{\lambda}y})$ correction at the saddle point. The function has the same branch cut as the leading order. Since the saddle point is located on the second sheet we need to analytically continue the function $\mathcal{F}_1^S(e^{\sqrt{\lambda}y})$.

To do this, recall that we found that the subleading free energy can be expressed in terms of leading order potentials \eqref{eq:tpartialtF1}. The continuation $ \frac{\dd}{\dd y} \F_1^{S, \rm II}$ to the second sheet should satisfy
\begin{equation}
    \frac{\dd}{\dd y} \F^{S, \rm II}_1(e^{\sqrt{\lambda}y})= 2\omega^{S,\rm II}_0(e^{\sqrt{\lambda}y})\nu^{S,\rm II}_0(e^{\sqrt{\lambda}y}).
\end{equation}
The analytic continuation of the leading order potentials to the second sheet is
\begin{align}
    &\omega_0^{S, \rm II}(e^{\sqrt{\lambda}y}) = \omega_0^{S, \rm I}(e^{\sqrt{\lambda}y}) - \frac{1}{\sqrt{y^2-1}}\\
    &\nu_0^{S, \rm II}(e^{\sqrt{\lambda}y}) = \nu_0^{S, \rm I}(e^{\sqrt{\lambda}y}) + \frac{y}{\sqrt{y^2-1}}.
\end{align}
For fixed $y<-1$, the first-sheet potentials are exponentially suppressed at strong coupling
\begin{equation}
    \omega_0^{S, \rm I}(e^{\sqrt{\lambda}y}) = \frac{\sqrt{\lambda}}{2\pi}\int_{-1}^1 \frac{\dd x}{\sqrt{1-x^2}}\frac{e^{\sqrt{\lambda}(x+y)}}{1-e^{\sqrt{\lambda}(x+y)}}
\end{equation}
because $x+y\leq1+y<0$ throughout the integration region. The integrand for $\nu_0^{S, \rm I}$ differs only by an additional factor of $x$ which does not change the asymptotic behaviour. Neglecting these exponentially suppressed terms, we obtain
\begin{equation}
    \frac{\dd}{\dd y} \F_1^{S, \rm II}(e^{\sqrt{\lambda}y}) = - \frac{2y}{y^2-1}.
\end{equation}
Integrating with respect to $y$ we obtain up to exponentially suppressed terms
\begin{equation}\label{eq:withc(g)}
    \F_1^{S, \rm II}(e^{\sqrt{\lambda}y}) = -\log(y^2-1)+c(\lambda).
\end{equation}
We have described the analytic continuation in this simplified form, using the differential relation to leading-order potentials. We could have also examined the full integral \eqref{eq:F1fullintegral}, which would then fix the constant as well:
\begin{equation}
    c(\lambda)=-\log \sqrt{\lambda}+{\rm const}.
\end{equation}
For example, this could be obtained by inspecting the behaviour of \eqref{eq:F1fullintegral} around $y=-1$ and comparing with \eqref{eq:withc(g)}.
Using $y_*^2-1 = \kappa^2$, the subleading free energy at the saddle point is
\begin{equation}\label{eq: F1atsaddle}
    \F_1^{S, \rm II}(e^{\sqrt{\lambda}y_*}) = -2\log \kappa -\log \sqrt{\lambda}+{\rm const}+ \cO(e^{-\sqrt{\lambda}}).
\end{equation}

The subleading correction to $W_{S_k}$ also receives contributions from one-loop fluctuations around the saddle point and from the measure of the contour integral. The one-loop fluctuations around the saddle point give
\begin{equation}
    -\dfrac{1}{2}\log\left(-\dfrac{\dd^2 \Phi_0^{\rm II}}{\dd y^2}(y_*)\right) = \dfrac{1}{2}\log\left( \dfrac{\kappa}{4\sqrt{\kappa^2+1}} \right).
\end{equation}
In addition, the integral measure gives a contribution
\begin{equation}
    \dfrac{\dd t}{t}=e^{\log \sqrt{\lambda}}\dd y,
\end{equation}
which cancels the $\lambda$-dependence in \eqref{eq: F1atsaddle}.
Collecting all contributions gives
\begin{equation}
    W_{S_k}= \exp\left( N\Phi_0^{\rm II}(y_*)+\F_1^{S,\rm II}(e^{\sqrt{\lambda}y_*}) -\dfrac{1}{2}\log\left(-\dfrac{\dd^2 \Phi_0^{\rm II}}{\dd y^2}(y_*)\right)+\log \sqrt{\lambda}+{\rm const}\right).
\end{equation}
Evaluating the terms we obtain 
\begin{equation}
\begin{split}
    \log W_{S_k} &= 2N\lb \kappa \sqrt{1+\kappa^2}+ {\rm arcsinh}\, \kappa \rb-\frac12 \log\Big(\kappa^3 \sqrt{1+\kappa^2}\Big) + \operatorname{const}+ \cO(e^{-\sqrt{\lambda}}).
\end{split}
\end{equation}
This agrees with the $\kappa$-dependence in \cite{Chen-Lin:2016kkk}\footnote{In~\cite{Chen-Lin:2016kkk}, the Wilson loop is defined as
$\log W_{S_k}=-N(\mathcal F_0+N^{-1}\mathcal F_1+\cdots)$, giving the difference in sign.}

\section{Conclusions}

In this paper we considered the subleading order correction to the large-$N$, large-$k$ asymptotic expansion of Wilson loops in the rank-$k$ symmetric and antisymmetric representations of $U(N)$ in $\mathcal{N}=4$ SYM theory, with $f=\frac kN$ held fixed. There are persistent discrepancies between the supergravity calculations for the giant Wilson loops and the matrix model analysis. 
With this calculation we hope to nail down the matrix model side.

Our approach is based on a Fredholm determinant representation of the generating functions of the Wilson loops. Its main advantage is that the antisymmetric and symmetric cases can be treated in the same way, with the distinction between them encoded in the symbol function. This approach, together with the Tracy--Widom method of differential equations, allowed us to systematically compute subleading corrections in $1/N$ using the conventional WKB expansion. In particular, we obtained the leading and first subleading contributions to the free energies at finite coupling for both representations, as well as a formula for the first subleading term for a general symbol function \eqref{eq:general-F1}.

In the antisymmetric case, our results are in perfect agreement with previous calculations \cite{Gordon:2017dvy,CanazasGaray:2018cpk}. In the symmetric case, our result confirms that at strong coupling and fixed $\kappa=\frac{\sqrt{\lambda}f}{4}$, the subleading contribution to the rank-$k$ symmetric Wilson loop is the same as for the $k$-fold wound Wilson loop in the fundamental representation putting earlier result of \cite{Chen-Lin:2016kkk} on a firmer ground. 

All in all we conclude that at subleading order the gauge-theory calculations do not match the current predictions obtained from holography either for the D3-brane or for the D5-brane. There could be several sources for the mismatch from the holographic side such as gravitational backreaction of the probe branes or issues with the renormalization of one-loop determinants for D-brane fluctuations. 

Curiously, the mismatch is rather different in the two cases. For the D5-brane holographic dual of the antisymmetric Wilson loop, the subleading contribution from gauge theory scales as $\cO(N^0\lambda)$ at strong coupling. On the other hand, the semiclassical expansion of the DBI action around the classical D5-brane solution \cite{Yamaguchi:2006tq} is organized in powers of the D5-brane tension. The subleading term would therefore be expected to scale as $\cO(N^0\lambda^0)$ \cite{Faraggi:2011ge}. For the symmetric representation, the scaling of the subleading term from gauge theory agrees with that predicted by the semiclassical expansion around a classical D3-brane solution \cite{Drukker:2005kx}. The mismatch seems to instead come from the one-loop determinant contribution \cite{Buchbinder:2014nia, Chen-Lin:2016kkk}. It would be interesting to understand what is the nature of the mismatch, which we believe should be seeked on the holographic side of the story.

Our results place giant Wilson loops within the general framework of Fredholm determinants with integrable kernels, which has proved to be useful for several observables in $\mathcal{N}=2$ and $\mathcal{N}=4$ SYM. For the Bessel kernel, the choice of symbol function corresponds to different observables \cite{Bajnok:2024epf,Bajnok:2024bqr}. Since our calculation works for a generic symbol function $\chi\left(te^{\sqrt{\lambda}x}\right)$, it would be interesting to see if it can be extended to other physically interesting observables that admit a similar determinant representation. 

The closed system of integro-differential equations \eqref{alldiffeqs} also provides a starting point for computing higher-order corrections in $1/N$. Here we solved it through the first subleading order using the WKB expansion. A natural next step is to use the next WKB term from the recursion \eqref{eq:general-recursion} to determine $\F^{A,S}_2(t)$ and the corresponding giant Wilson loop corrections. In particular, it would be interesting to see if the massive cancellation we observed at the first subleading order persists at higher orders in $1/N$.
~\\

\noindent\textbf{Acknowledgements}\par\smallskip
\noindent
We would like to thank G. Korchemsky and A. Tseytlin for enlightening discussions. The authors used GPT-5.6 and GPT-6 Astra for assistance with some of the derivations and to check the paper for typos and errors. In particular it was used to guess the simplifications in  Appendix \ref{app:WKBdetails} and to bring the integral transformation \ref{lem:integraltransform} into its most general form. In cases where AI was used the derivations were repeated by the authors and they
take full responsibility for the contents of the paper.

\appendix
\section{Differential equations}
\label{app:diffeq}
We derive the differential equations satisfied by $Q(x,t)$ and $P(x,t)$ and the potentials $u(t), w(t)$ and $v(t)$. For notational simplicity, we suppress the $t$-dependence. Recall
\begin{equation}
    Q(x)=\bra{x}\frac{1}{1+\bm{H}\opchi}\ket{\phi}, 
    \qquad 
    P(x)=\bra{x}\frac{1}{1+\bm{H}\opchi}\ket{\psi}.
\end{equation}

Differentiating $Q(x)$ and using the resolvent commutator identity gives
\begin{equation}\label{eq:Qprime}
    Q'(x)=\bra{x}\frac{1}{1+\bm{H}\opchi}\partial_x\ket{\phi}
    +\bra{x}\Big[\partial_x,\frac{1}{1+\bm{H}\opchi}\Big]\ket{\phi}.
\end{equation}
We treat the two terms in \eqref{eq:Qprime} separately. For the first term, using $\phi'(x)=-2Nx\phi(x)+2N\psi(x)$ and the identity
\begin{equation}
    \Big[\boldsymbol{x},\frac{1}{1+\bm{H}\opchi}\Big]
    =-\frac{1}{1+\bm{H}\opchi}[\boldsymbol{x},\bm{H}]\opchi
    \frac{1}{1+\bm{H}\opchi},
\end{equation}
together with $[\boldsymbol{x},\bm{H}]=\ket{\phi}\bra{\psi}-\ket{\psi}\bra{\phi}$, we obtain
\begin{equation}
    \bra{x}\frac{1}{1+\bm{H}\opchi}\partial_x\ket{\phi}
    =2N\big(-xQ(x)+P(x)-Q(x)v+P(x)u\big).
\end{equation}

For the second term we use
\begin{equation}
    \Big[\partial_x,\frac{1}{1+\bm{H}\opchi}\Big]
    =-\frac{1}{1+\bm{H}\opchi}[\partial_x,\bm{H}]
    \opchi\frac{1}{1+\bm{H}\opchi}
    -\frac{1}{1+\bm{H}\opchi}\bm{H}(\partial_x\opchi)
    \frac{1}{1+\bm{H}\opchi}.
\end{equation}
The expression for the kernel of  $[\partial_x,\bm{H}]$ in \eqref{eq:hermiteDerivativecommutator} can be written in operator notation as 
\begin{equation}
      [\partial_x,\bm{H}] = -2N \ket{\phi}\bra{\psi}- 2N \ket{\psi} \bra{\phi}.
\end{equation}
Evaluating the matrix elements gives
\begin{equation}\label{eq:appendixA1}
    \bra{x}\Big[\partial_x,\frac{1}{1+\bm{H}\opchi}\Big]\ket{\phi}
    =2N\big(Q(x)v+P(x)u\big)
    -\bra{x}\frac{1}{1+\bm{H}\opchi}\bm{H}(\partial_x\opchi)
    \frac{1}{1+\bm{H}\opchi}\ket{\phi}.
\end{equation}
Here the potentials \eqref{eq:potentials} have appeared. In order to simplify the second term in \eqref{eq:appendixA1} we have to require that the $x$-derivative of the symbol function can be traded for the $t$-derivative.
Using the symbol function identity \eqref{symbolfunctionidentity},
\begin{equation}
    \partial_x\chi(x,t)=\sqrt{\lambda}\,t\,\partial_t\chi(x,t),
\end{equation}
the remaining term becomes
\begin{equation}
    \bra{x}\frac{1}{1+\bm{H}\opchi}\bm{H}(\partial_x\opchi)
    \frac{1}{1+\bm{H}\opchi}\ket{\phi}
    =-\sqrt{\lambda}\,t\,\partial_t Q(x).
\end{equation}
Collecting all terms together we obtain
\begin{equation}
    Q'(x)=2N\big[-xQ(x)+P(x)+2P(x)u\big]
    +\sqrt{\lambda}\,t\,\partial_t Q(x).
    \label{QdiffeqDerivation}
\end{equation}

For $P(x)$, we proceed analogously
\begin{equation}
    P'(x)=\bra{x}\frac{1}{1+\bm{H}\opchi}\partial_x\ket{\psi}
    +\bra{x}\Big[\partial_x,\frac{1}{1+\bm{H}\opchi}\Big]\ket{\psi}.
\end{equation}
Using $\psi'(x)=-2N\phi(x)+2Nx\psi(x)$ and the commutator relations above gives
\begin{equation}
    P'(x)=2N\big[xP(x)-Q(x)+2Q(x)w\big]
    +\sqrt{\lambda}\,t\,\partial_t P(x).
    \label{PdiffeqDerivation}
\end{equation}

Differentiating $u=\bra{\phi}\opchi(1+\bm{H}\opchi)^{-1}\ket{\phi}$ and using the resolvent identity yields
\begin{align}
    \partial_t u
    &=\bra{\phi}\frac{1}{1+\opchi\bm{H}}
      (\partial_t\opchi)
      \frac{1}{1+\bm{H}\opchi}\ket{\phi} \\
    &=\int_{-\infty}^{\infty}\dd x\,Q^2(x)\partial_t\chi(x,t).
    \label{udiffeq}
\end{align}
The same computation gives
\begin{align}
    &\partial_t w
    =\int_{-\infty}^{\infty}\dd x\,P^2(x)\partial_t\chi(x,t)
    \label{sdiffeq}\\
        &\partial_t v
    =\int_{-\infty}^{\infty}\dd x\,Q(x)P(x)\partial_t\chi(x,t).
    \label{vdiffeq}
\end{align}

Next we derive the two relations \eqref{eq:vFrelation} and \eqref{eq:uwvrelation}. The function $R(x,y)\chi(y)$ is the kernel of the operator
\begin{equation}
  R(x,y)\chi(y) = \bra{x}(1+\bm{H}\opchi)^{-1}\bm{H}\opchi\ket{y} = \bra{x}1- (1+\bm{H}\opchi)^{-1}\ket{y}.
\end{equation}
Using the resolvent identity and the relations \eqref{eq:hermiteDerivativecommutator}--\eqref{symbolfunctionidentity}, we find
\begin{equation}
    (\partial_x+\partial_y)\big(R(x,y)\chi(y)\big)
    = -2NQ(x)P(y)\chi(y)-2NP(x)Q(y)\chi(y)
      +\sqrt{\lambda}\,t\partial_t\big(R(x,y)\chi(y)\big).
\end{equation}
Taking $x\to y$ and using $\D\chi = 0$ gives
\begin{equation}
    (\partial_x-\sqrt{\lambda}\,t\partial_t)R(x,x) = -4NQ(x)P(x).
    \label{eq:Rxeq}
\end{equation}
Multiplying \eqref{eq:Rxeq} by $\partial_t\chi(x)$ and integrating over the real line, the right-hand side gives $-4N\partial_t v$ via \eqref{vdiffeq}. On the left-hand side, using integration by parts and $\D\chi = 0$ gives
\begin{equation}
    \partial_t \left(-4v+\sqrt{\lambda}\,t\partial_t\F(t)\right) = 0.
\end{equation}
Since $v(t=0)=0$, this integrates to \eqref{eq:vFrelation}.

For the second relation, differentiate $2uw-u+w$ with respect to $t$.
Using \eqref{alldiffeqs},
\begin{equation}
    \partial_t(2uw-u+w)
    = \int\dd x\,\big[(2w-1)Q^2+(2u+1)P^2\big]\,\partial_t\chi
    = \frac{1}{2N}\int\dd x\,\D\big(QP\big)\,\partial_t\chi.
\end{equation}
Applying the same integration-by-parts as above gives \eqref{eq:uwvrelation}.

\section{Subleading WKB details}
\label{app:WKBdetails}
In this appendix we show that $R_1(z)= -R_2(z)$. Recall that we use the notations \eqref{eq:notations}. The function $R_1(z)$ was defined in \eqref{eq:R1-before-finite-part} and $R_2(z)$ was defined in \eqref{eq:R2-before-moments}.

We first simplify both expressions. Starting with $R_2(z)$, we can evaluate the $s$-integral along the cut using the fact that the leading-order definitions of the two potentials imply
  \begin{align}
  \int_{-1}^{1}
  \frac{n \left(\tau_{y-s}(z)\right)}
  {\sqrt{1-s^2}}\dd s
  &=
  \frac{2\pi}{\sqrt{\lambda}}\,
  \omega_0 \left(\tau_{y}(z)\right),
  \label{eq:cut-moment-zero}\\
  \int_{-1}^{1}
  \frac{s\,n \left(\tau_{y-s}(z)\right)}
  {\sqrt{1-s^2}}\dd s
  &=
  \frac{2\pi}{\sqrt{\lambda}}\,
  \nu_0 \left(\tau_{y}(z)\right).
  \label{eq:cut-moment-one}
  \end{align}
  Moreover, 
  \begin{equation}
  \partial_s
  n \left(\tau_{y-s}(z)\right)
  =
  \sqrt{\lambda}\,
  \theta n \left(\tau_{y-s}(z)\right),
  \end{equation}
  where $\theta f(\zeta)=\zeta\partial_\zeta f(\zeta)$.
  Using the cut finite-part identity
  \begin{equation}
   \int_{-1}^{1}
  \frac{f(s)}{(1-s^2)^{3/2}}\dd s
  =
  -\int_{-1}^{1}
  \frac{s f'(s)}{\sqrt{1-s^2}}\dd s,
  \end{equation}
  we obtain
  \begin{align}
   \int_{-1}^{1}
  \frac{
  n\left(\tau_{y-s}(z)\right)}
  {(1-s^2)^{3/2}}\dd s
  &=
  -2\pi\,
 \theta\nu_0 \left(\tau_{y}(z)\right),
  \label{eq:fp-moment-zero}
  \\
   \int_{-1}^{1}
  \frac{
  s\,n \left(\tau_{y-s}(z)\right)}
  {(1-s^2)^{3/2}}\dd s
  &=
  -2\pi\,
  \theta\omega_0 \left(\tau_{y}(z)\right),
  \label{eq:fp-moment-one}
  \\
   \int_{-1}^{1}
  \frac{
  s^2\,n\left(\tau_{y-s}(z)\right)}
  {(1-s^2)^{3/2}}\dd s
  &=
  -2\pi\,
  \theta\nu_0 \left(\tau_{y}(z)\right)
  -\frac{2\pi}{\sqrt{\lambda}}\,
  \omega_0 \left(\tau_{y}(z)\right).
  \label{eq:fp-moment-two}
  \end{align}
  Substituting
  \eqref{eq:cut-moment-zero}--\eqref{eq:fp-moment-two} into
  \eqref{eq:R2-before-moments} gives
  \begin{align}
  R_2(z)
  ={}&
  \sqrt{\lambda}
  \int_1^\infty
  \frac{\dd y}{\sqrt{y^2-1}}
  \Bigg\{
  \left[
  y\left(
  \frac12+
  \nu_0 \left(\tau_{y}(z)\right)
  \right)
  +
  \omega_0 \left(\tau_{y}(z)\right)
  \right]
   \theta\nu_0 \left(\tau_{y}(z)\right)
  \nonumber\\
  &\hspace{34mm}
  +
  \left[
  y\,\omega_0 \left(\tau_{y}(z)\right)
  +
  \nu_0 \left(\tau_{y}(z)\right)
  \right]
   \theta\omega_0 \left(\tau_{y}(z)\right)
  \Bigg\}.
  \label{eq:R2-final}
  \end{align}

Next we simplify the $R_1(z)$ integral by considering the identity
\begin{equation}
   \int_1^\infty
  \frac{f(y)}{(y^2-1)^{3/2}}\dd y
  =
  \int_1^\infty
  \frac{y f'(y)}{\sqrt{y^2-1}}\dd y
  \label{eq:exterior-finite-part-identity}
  \end{equation}
  and 
  \begin{align}
  \frac{\dd}{\dd y}
  \omega_0 \left(\tau_{y}(z)\right)
  &=
  -\sqrt{\lambda}\,
  \theta\omega_0 \left(\tau_{y}(z)\right),
  \\
  \frac{\dd}{\dd y}
  \nu_0 \left(\tau_{y}(z)\right)
  &=
  -\sqrt{\lambda}\,
  \theta\nu_0 \left(\tau_{y}(z)\right).
  \end{align}
  Using \eqref{eq:exterior-finite-part-identity} in
  \eqref{eq:R1-before-finite-part} gives
  \begin{align}
  R_1(z)
  ={}&
  \frac12
  \int_1^\infty
  \frac{\dd y}{\sqrt{y^2-1}}
  \Bigg\{
  y\Bigg[
  2\left(
  y\,\omega_0 \left(\tau_{y}(z)\right)
  +
  \nu_0 \left(\tau_{y}(z)\right)
  \right)
  +1
  \Bigg]
  \nonumber\\
  &\quad\times
  \Bigg[
  \omega_0 \left(\tau_{y}(z)\right)
  -\sqrt{\lambda}
  \Big(
  y\,\theta\omega_0 \left(\tau_{y}(z)\right)
  +
  \theta\nu_0 \left(\tau_{y}(z)\right)
  \Big)
  \Bigg]+
  \sqrt{\lambda}\,
  \theta\omega_0 \left(\tau_{y}(z)\right)
  -
  \omega_0 \left(\tau_{y}(z)\right)^2
  \Bigg\}.
  \label{eq:R1-final}
  \end{align}
It is possible to show that the integrand in $R_1(z)+R_2(z)$ is a total derivative,
\begin{align}
    R_1(z)+R_2(z) &= \int_1^\infty \dd y\frac{\dd}{\dd y}
  \left[
  \sqrt{y^2-1}
  \Bigg(
  \frac12\omega_0 \left(\tau_{y}(z)\right)
  +\frac{y}{2}\omega_0 \left(\tau_{y}(z)\right)^2
  +\nu_0 \left(\tau_{y}(z)\right)
   \omega_0 \left(\tau_{y}(z)\right)
  \Bigg)  \right].
\end{align}
Consequently,
  \begin{align}
  R_1(z)+R_2(z)
  ={}&
  \Bigg[
  \sqrt{y^2-1}
  \Bigg(
  \frac12\omega_0 \left(\tau_{y}(z)\right)
  +\frac{y}{2}\omega_0 \left(\tau_{y}(z)\right)^2
  +\nu_0 \left(\tau_{y}(z)\right)
   \omega_0 \left(\tau_{y}(z)\right)
  \Bigg)
  \Bigg]_{y=1}^{y=\infty}.
  \end{align}

  At $y=1$, this expression vanishes because
  $\sqrt{y^2-1}=0$, while the leading potentials are finite.  At
  $y\to\infty$, the argument $\tau_{y}(z)=z e^{-\sqrt{\lambda}y}$ tends to zero,
  and both $\omega_0$ and $\nu_0$ vanish.  Therefore,
  \begin{equation}
  R_1(z)=-R_2(z).
  \label{eq:R1=R2-derivation}
  \end{equation}

\section{Integral transform}
\label{sectionIntegralTransform}
The leading- and subleading-order WKB equations involve the same integral transform. In this appendix we define this transform and prove the identity that allows us to solve the integral equations for $G_0(z)$ and $G_1(z)$. 
\begin{lem}
\label{lem:integraltransform} 
Define an integral transform $\mathsf{T}$ of a function $F(z)$, such that $F(0)=0$,
\begin{equation}
    \mathsf{T}[F](z) = \sqrt{\lambda}\int_1^\infty\dd y\int_{-1}^1\dd s\,
    \,F \left(ze^{\sqrt{\lambda}(s-y)}\right)
    \frac{s+y}{\sqrt{1-s^2}\sqrt{y^2-1}}.
\end{equation}
Then 
\begin{equation}
   \dfrac{\dd}{\dd z} \mathsf{T}[F](z) = \dfrac{\pi}{z}F(z). \label{integraltransformidentity-app}
\end{equation}
\end{lem}

\begin{proof}
Change variables to $\mu = y-s$. Then
\begin{equation}
    \mathsf{T}[F](z) = \sqrt{\lambda} \int_0^\infty \dd \mu F\left(ze^{-\sqrt{\lambda}\mu}\right) \int_{\max(1-\mu,-1)}^1\dd s\,\frac{2s+\mu}{\sqrt{1-s^2}\sqrt{(s+\mu)^2-1}}.
\end{equation}

To evaluate the inner $s$-integral, let $\tau = s^2 + \mu s -1$. Then the integrand simplifies to
\begin{equation}
    \dd s \frac{2s+\mu}{\sqrt{(1-s^2)\big((s+\mu)^2-1\big)}}=\frac{\dd \tau}{\sqrt{\mu^2-\tau^2}}.
\end{equation}
The integration region is $\tau \in [-\mu,\mu]$, giving
\begin{equation}
    \int_{-\mu}^\mu \frac{\dd \tau}{\sqrt{\mu^2-\tau^2}} = \pi.
\end{equation}
Therefore,
\begin{equation}
    \mathsf{T}[F](z) = \pi \sqrt{\lambda}\int_0^\infty\dd \mu\,F \left(ze^{-\sqrt{\lambda}\mu}\right).
\end{equation}
Differentiating and setting $q=ze^{-\sqrt{\lambda}\mu}$ gives
\begin{equation}
    \frac{\dd}{\dd z}\mathsf{T}[F](z) = \frac{\pi}{z}\int_0^z \dd qF'(q) = \frac{\pi}{z}\Big( F(z) -F(0)\Big)=\frac{\pi}{z}F(z).
\end{equation}
\end{proof}

\bibliographystyle{JHEP}

\providecommand{\href}[2]{#2}\begingroup\raggedright\endgroup

\end{document}